\documentclass[10pt, conference, letterpaper]{IEEEtran}
\IEEEoverridecommandlockouts

\usepackage[english]{babel}

\usepackage{cite}
\usepackage{amsmath,amsfonts}
\usepackage{algorithmic,algorithm}
\usepackage{graphicx}
\usepackage{textcomp}
\usepackage{xcolor}
\usepackage{tikz}
\usetikzlibrary{patterns}
\usepackage{caption}
\usepackage{subcaption}
\usepackage{cleveref}

\newcommand\algorithmicprocedure{\textbf{procedure}}
\newcommand{\algorithmicendprocedure}{\algorithmicend\ \algorithmicprocedure}
\makeatletter
\newcommand\PROCEDURE[3][default]{%
	\ALC@it
	\algorithmicprocedure\ \textsc{#2}(#3)%
	\ALC@com{#1}%
	\begin{ALC@prc}%
	}
	\newcommand\ENDPROCEDURE{%
	\end{ALC@prc}%
	\ifthenelse{\boolean{ALC@noend}}{}{%
		\ALC@it\algorithmicendprocedure
	}%
}
\newenvironment{ALC@prc}{\begin{ALC@g}}{\end{ALC@g}}
\makeatother

\newtheorem{remark}{Remark}
\newtheorem{definition}{Definition}
\newtheorem{assumption}{Assumption}

\newtheorem{lemma}{Lemma}
\newtheorem{theorem}{Theorem}

\def\avA{\overline{A}}
\def\avL{\overline{L}}
\def\avR{\overline{R}}
\def\avT{\overline{T}}
\def\avdelta{\overline{(\delta_\ell^\pi)^2}}

\newenvironment{proof}{\begin{IEEEproof}}{\end{IEEEproof}}

\def\BibTeX{{\rm B\kern-.05em{\sc i\kern-.025em b}\kern-.08em
		T\kern-.1667em\lower.7ex\hbox{E}\kern-.125emX}}

\begin{document}
\title{Minimizing Age of Information under Latency and Throughput Constraints} 

\author{\IEEEauthorblockN{Kumar Saurav}
	\IEEEauthorblockA{
		\textit{Tata Institute of Fundamental Research}\\
		Mumbai, India. \\
		kumar124103@nitp.ac.in}
	\and
	\IEEEauthorblockN{Ashok Kumar Reddy Chavva}
	\IEEEauthorblockA{\textit{Samsung R\&D Institute India-Bangalore} \\
		\textit{Bengaluru, India}\\
		ashok.chavva@samsung.com}
}
\vspace{-2in}

\def\onehalf{\frac{1}{2}}
\def\etal{et.\/ al.\/}
\newcommand{\bydef}{\triangleq}
\newcommand{\tr}{{\it{tr}}}
\def\SNR{{\textsf{SNR}}}
\def\bydef{:=}
\def\bba{{\mathbb{a}}}
\def\bbb{{\mathbb{b}}}
\def\bbc{{\mathbb{c}}}
\def\bbd{{\mathbb{d}}}
\def\bbee{{\mathbb{e}}}
\def\bbff{{\mathbb{f}}}
\def\bbg{{\mathbb{g}}}
\def\bbh{{\mathbb{h}}}
\def\bbi{{\mathbb{i}}}
\def\bbj{{\mathbb{j}}}
\def\bbk{{\mathbb{k}}}
\def\bbl{{\mathbb{l}}}
\def\bbm{{\mathbb{m}}}
\def\bbn{{\mathbb{n}}}
\def\bbo{{\mathbb{o}}}
\def\bbp{{\mathbb{p}}}
\def\bbq{{\mathbb{q}}}
\def\bbr{{\mathbb{r}}}
\def\bbs{{\mathbb{s}}}
\def\bbt{{\mathbb{t}}}
\def\bbu{{\mathbb{u}}}
\def\bbv{{\mathbb{v}}}
\def\bbw{{\mathbb{w}}}
\def\bbx{{\mathbb{x}}}
\def\bby{{\mathbb{y}}}
\def\bbz{{\mathbb{z}}}
\def\bb0{{\mathbb{0}}}

\def\bydef{:=}
\def\ba{{\mathbf{a}}}
\def\bb{{\mathbf{b}}}
\def\bc{{\mathbf{c}}}
\def\bd{{\mathbf{d}}}
\def\bee{{\mathbf{e}}}
\def\bff{{\mathbf{f}}}
\def\bg{{\mathbf{g}}}
\def\bh{{\mathbf{h}}}
\def\bi{{\mathbf{i}}}
\def\bj{{\mathbf{j}}}
\def\bk{{\mathbf{k}}}
\def\bl{{\mathbf{l}}}
\def\bm{{\mathbf{m}}}
\def\bn{{\mathbf{n}}}
\def\bo{{\mathbf{o}}}
\def\bp{{\mathbf{p}}}
\def\bq{{\mathbf{q}}}
\def\br{{\mathbf{r}}}
\def\bs{{\mathbf{s}}}
\def\bt{{\mathbf{t}}}
\def\bu{{\mathbf{u}}}
\def\bv{{\mathbf{v}}}
\def\bw{{\mathbf{w}}}
\def\bx{{\mathbf{x}}}
\def\by{{\mathbf{y}}}
\def\bz{{\mathbf{z}}}
\def\b0{{\mathbf{0}}}
\def\opt{\mathsf{OPT}}
\def\on{\mathsf{ON}}
\def\off{\mathsf{OFF}}
\def\bA{{\mathbf{A}}}
\def\bB{{\mathbf{B}}}
\def\bC{{\mathbf{C}}}
\def\bD{{\mathbf{D}}}
\def\bE{{\mathbf{E}}}
\def\bF{{\mathbf{F}}}
\def\bG{{\mathbf{G}}}
\def\bH{{\mathbf{H}}}
\def\bI{{\mathbf{I}}}
\def\bJ{{\mathbf{J}}}
\def\bK{{\mathbf{K}}}
\def\bL{{\mathbf{L}}}
\def\bM{{\mathbf{M}}}
\def\bN{{\mathbf{N}}}
\def\bO{{\mathbf{O}}}
\def\bP{{\mathbf{P}}}
\def\bQ{{\mathbf{Q}}}
\def\bR{{\mathbf{R}}}
\def\bS{{\mathbf{S}}}
\def\bT{{\mathbf{T}}}
\def\bU{{\mathbf{U}}}
\def\bV{{\mathbf{V}}}
\def\bW{{\mathbf{W}}}
\def\bX{{\mathbf{X}}}
\def\bY{{\mathbf{Y}}}
\def\bZ{{\mathbf{Z}}}
\def\b1{{\mathbf{1}}}

\def\bbA{{\mathbb{A}}}
\def\bbB{{\mathbb{B}}}
\def\bbC{{\mathbb{C}}}
\def\bbD{{\mathbb{D}}}
\def\bbE{{\mathbb{E}}}
\def\bbF{{\mathbb{F}}}
\def\bbG{{\mathbb{G}}}
\def\bbH{{\mathbb{H}}}
\def\bbI{{\mathbb{I}}}
\def\bbJ{{\mathbb{J}}}
\def\bbK{{\mathbb{K}}}
\def\bbL{{\mathbb{L}}}
\def\bbM{{\mathbb{M}}}
\def\bbN{{\mathbb{N}}}
\def\bbO{{\mathbb{O}}}
\def\bbP{{\mathbb{P}}}
\def\bbQ{{\mathbb{Q}}}
\def\bbR{{\mathbb{R}}}
\def\bbS{{\mathbb{S}}}
\def\bbT{{\mathbb{T}}}
\def\bbU{{\mathbb{U}}}
\def\bbV{{\mathbb{V}}}
\def\bbW{{\mathbb{W}}}
\def\bbX{{\mathbb{X}}}
\def\bbY{{\mathbb{Y}}}
\def\bbZ{{\mathbb{Z}}}

\def\cA{\mathcal{A}}
\def\cB{\mathcal{B}}
\def\cC{\mathcal{C}}
\def\cD{\mathcal{D}}
\def\cE{\mathcal{E}}
\def\cF{\mathcal{F}}
\def\cG{\mathcal{G}}
\def\cH{\mathcal{H}}
\def\cI{\mathcal{I}}
\def\cJ{\mathcal{J}}
\def\cK{\mathcal{K}}
\def\cL{\mathcal{L}}
\def\cM{\mathcal{M}}
\def\cN{\mathcal{N}}
\def\cO{\mathcal{O}}
\def\cP{\mathcal{P}}
\def\cQ{\mathcal{Q}}
\def\cR{\mathcal{R}}
\def\cS{\mathcal{S}}
\def\cT{\mathcal{T}}
\def\cU{\mathcal{U}}
\def\cV{\mathcal{V}}
\def\cW{\mathcal{W}}
\def\cX{\mathcal{X}}
\def\cY{\mathcal{Y}}
\def\cZ{\mathcal{Z}}

\def\sfA{\mathsf{A}}
\def\sfB{\mathsf{B}}
\def\sfC{\mathsf{C}}
\def\sfD{\mathsf{D}}
\def\sfE{\mathsf{E}}
\def\sfF{\mathsf{F}}
\def\sfG{\mathsf{G}}
\def\sfH{\mathsf{H}}
\def\sfI{\mathsf{I}}
\def\sfJ{\mathsf{J}}
\def\sfK{\mathsf{K}}
\def\sfL{\mathsf{L}}
\def\sfM{\mathsf{M}}
\def\sfN{\mathsf{N}}
\def\sfO{\mathsf{O}}
\def\sfP{\mathsf{P}}
\def\sfQ{\mathsf{Q}}
\def\sfR{\mathsf{R}}
\def\sfS{\mathsf{S}}
\def\sfT{\mathsf{T}}
\def\sfU{\mathsf{U}}
\def\sfV{\mathsf{V}}
\def\sfW{\mathsf{W}}
\def\sfX{\mathsf{X}}
\def\sfY{\mathsf{Y}}
\def\sfZ{\mathsf{Z}}

\def\bydef{:=}
\def\sfa{{\mathsf{a}}}
\def\sfb{{\mathsf{b}}}
\def\sfc{{\mathsf{c}}}
\def\sfd{{\mathsf{d}}}
\def\sfee{{\mathsf{e}}}
\def\sfff{{\mathsf{f}}}
\def\sfg{{\mathsf{g}}}
\def\sfh{{\mathsf{h}}}
\def\sfi{{\mathsf{i}}}
\def\sfj{{\mathsf{j}}}
\def\sfk{{\mathsf{k}}}
\def\sfl{{\mathsf{l}}}
\def\sfm{{\mathsf{m}}}
\def\sfn{{\mathsf{n}}}
\def\sfo{{\mathsf{o}}}
\def\sfp{{\mathsf{p}}}
\def\sfq{{\mathsf{q}}}
\def\sfr{{\mathsf{r}}}
\def\sfs{{\mathsf{s}}}
\def\sft{{\mathsf{t}}}
\def\sfu{{\mathsf{u}}}
\def\sfv{{\mathsf{v}}}
\def\sfw{{\mathsf{w}}}
\def\sfx{{\mathsf{x}}}
\def\sfy{{\mathsf{y}}}
\def\sfz{{\mathsf{z}}}
\def\sf0{{\mathsf{0}}}

\def\Nt{{N_t}}
\def\Nr{{N_r}}
\def\Ne{{N_e}}
\def\Ns{{N_s}}
\def\Es{{E_s}}
\def\No{{N_o}}
\def\sinc{\mathrm{sinc}}
\def\dmin{d^2_{\mathrm{min}}}
\def\vec{\mathrm{vec}~}
\def\kron{\otimes}
\def\Pe{{P_e}}
\newcommand{\expeq}{\stackrel{.}{=}}
\newcommand{\expg}{\stackrel{.}{\ge}}
\newcommand{\expl}{\stackrel{.}{\le}}
\def\SIR{{\mathsf{SIR}}}

\def\nn{\nonumber}

\maketitle
\begin{abstract}
    We consider a scheduling problem pertinent to a base station (BS) that serves heterogeneous users (UEs). We focus on a downlink model, with users that are either interested in large throughput, low latency (the time difference between the packet arrival at the BS, and reception at the UE), or low age-of-information (AoI, the difference between the current time and the arrival time of the latest packet at the BS that has been received at the UE). In each time step, the BS may serve a limited number of UEs. The objective is to find a scheduling algorithm that meets the individual service requirements of all the UEs. In this paper, we formulate the aforementioned objective as an optimization problem, and propose candidate solutions (algorithms) that rationalize the tradeoff between throughput, latency and the AoI. Then, we analyze the performance of the proposed algorithms both analytically and numerically, and discuss techniques to extend their applicability to other general setups.
\end{abstract}


\maketitle

\section{Introduction} \label{sec:introduction}

Scheduling problems have always been at the heart of communication/computer networks, especially with large number of devices (UEs) \cite{pantelidou2011scheduling}. However, most classical work on scheduling have focused on homogeneous networks, where all UEs have similar service requirements \cite{hou2009theory,lee2007reverse,buyukkoc1985cmu,saurav2022scheduling}. 
In this paper, we consider a scheduling problem pertinent to modern networks where UEs have heterogeneous service (QoS) requirements. 
For example, consider a private 5G network, where a base station (BS) serves three types of UEs:

1) Televisions, music systems, etc. that need to download large volumes of data. For these UEs, the BS needs to ensure high \emph{throughput} (information bits per unit time).

2) Robotic systems that need to be fed tasks and time-sensitive instructions. The BS needs to ensure that the \emph{latency} (i.e. the time-difference between when an instruction arrives at the BS, and when it is received at the UE) of the instruction packets for these UEs be minimized.

3) Air conditioners, security alarms, etc. that need to know the current state of the environment. The BS should periodically transmit information updates to these UEs such that at any time, their \emph{age of information} (AoI \cite{kaul2012status}; equal to the difference between the current time and the generation time of the latest update received at the UE) remains small.

\noindent Since UEs have diverse QoS requirements, and the BS can serve only a subset of UEs at a time, a critical question in such networks is: which UEs to prioritize (serve) at any time? 
In this paper, we address this question by studying the intricate tradeoffs between the three most common QoS metrics, i.e. throughput, latency and AoI, and developing a general framework to optimize them in a heterogeneous network. 
In particular, we consider an optimization problem where the objective is to minimize the weighted sum of the average AoI of a subset of UEs while satisfying the latency and throughput constraints of other subsets of UEs. Then we derive a general scheduling algorithm to solve this problem, and analyze its performance numerically.

\subsection*{Related Work}
For traditional QoS metrics i.e., throughput and latency, plenty of results are available across different network models \cite{mirzaeinnia2020latency,saxena2018review}. This includes works that have considered the metrics in isolation (e.g. \cite{kushner2004convergence}), as well as in conjuction (e.g. \cite{hou2009theory}). However, over the past decade, the traditional metrics have been found to be insufficient in capturing the needs of time-sensitive applications, especially that involves maintaining fresh data at network nodes. This includes applications such as remote estimation \cite{sun2017remote}, monitoring \cite{saurav2022scheduling}, control \cite{ayan2019age}, etc. For such applications, multiple metrics have been proposed such as age of information \cite{kaul2012status}, age of incorrect information \cite{maatouk2022age}, age of incorrect estimates \cite{joshi2021minimization}, value of information \cite{ayan2019age}, etc. Among all these metrics, age of information (AoI) has garnered most attention because of its generality, and ease in modeling real-world problems. However, despite the greater attention, work on some fundamental aspects of AoI are still under progress. This includes analyzing the inherent tradeoffs between AoI and other traditional metrics such as throughput and latency. 

In past, multiple attempts have been made towards understanding the tradeoff between AoI and throughput \cite{kadota2018optimizing,kosta2018age,saurav2021minimizing,saurav2022scheduling,gopal2019coexistence,gopal2018game,bhat2020throughput}. For example, \cite{kadota2018optimizing,kosta2018age} considered the problem of minimizing the weighted sum of AoI of UEs under minimum throughput constraint. Similarly, \cite{saurav2021minimizing,saurav2022scheduling} considered the problem of minimizing a linear combination of AoI and transmission cost, where transmission cost is proportional to throughput of UEs. In a decentralized setting, \cite{gopal2019coexistence,gopal2018game} analyzed the coexistence of AoI and throughput optimizing networks using game-theoretic approach. 

Tradeoff between AoI and latency is relatively less explored. Some relevant works are \cite{al2019information,sun2021age,fountoulakis2023scheduling}. In \cite{al2019information}, an optimization problem has been considered, where the goal is to minimize a linear combination of the average AoI of UEs and the probability of their latency (of transmitted packets) being greater than a threshold. Towards the goal, a randomized scheduling algorithm has been proposed, and analyzed numerically. While \cite{sun2021age,fountoulakis2023scheduling} considered a system where the goal is to minimize the average AoI of a subset of UEs, while ensuring that at least a certain fraction of all packets of another subset of UEs have latency less than a threshold. For this problem, \cite{sun2021age,fountoulakis2023scheduling} used the Lyapunov-drift plus penalty method to design a scheduling algorithms that is shown to satisfy the latency constraint. Also, upper bounds have been derived on the AoI cost, in terms of the system parameters. However, \cite{sun2021age,fountoulakis2023scheduling} have assumed generate-at-will model where a fresh packet  can be generated for AoI-sensitive UEs at any time, which is limiting in practical setups.
In this paper, we consider a general setup with stochastic packet arrivals, and conduct a comprehensive analysis of the three-way tradeoff between AoI, latency and throughput by solving the problem described next.

\section{System Model} \label{sec:SysModel}

Consider a base station (BS) that serves a set of users (UEs) $\cI$ (Figure \ref{fig:sysmodel}).  
Let the time be discretized into slots of unit length (henceforth, the terms `time' and `slot' are used interchangeably to refer to a slot $t\in\{1,2,3,\cdots\}$). 
At the start of each slot, for each UE $\ell$, a packet arrives at the BS with probability $q_\ell>0$, independently. We assume that the BS stores all the arriving packets until they are successfully transmitted (or intentionally discarded).
 
\begin{figure}[htbp]
	\centerline{\includegraphics[scale=0.5]{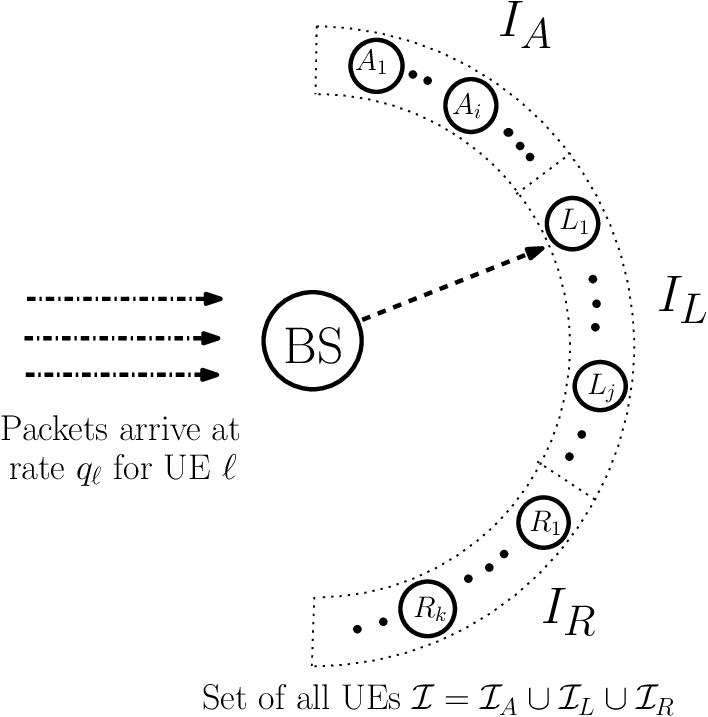}}
	\caption{System model.}
	\label{fig:sysmodel}
\end{figure}

Further, in each slot $t$, the BS can transmit at most one packet, to at most one UE. If the BS transmits a packet to UE $\ell$, then the transmission succeeds (the packet is received at UE $\ell$ by the end of slot $t$) with probability $p_\ell$, and fails otherwise. 

Based on the application, the UEs may have different quality-of-service (QoS) requirements. In this paper, we consider following three QoS metrics: 
\paragraph{Age of Information (AoI)} 
Let $g_{\ell i}$ denote the slot in which packet $i$ (for UE $\ell$) arrives at the BS. In any slot $t\ge g_{\ell i}$, the age of packet $i$ is equal to $t-g_{\ell i}$. Further, let $r_{\ell i}$ denote the slot in which the packet $i$ is received at UE $\ell$. In any slot $t$, the age of information (AoI; $A_\ell(t)$) of UE $\ell$ is defined as the minimum of the age of all packets received at UE $\ell$ by the end of slot $t-1$, i.e.
\begin{align} \label{eq:AoI}
	A_\ell(t)=\underset{\{i|r_{\ell i}\le t-1\}}{\min}\{t-g_{\ell i}\}=t-\lambda_\ell(t),
\end{align} 
where $\lambda_\ell(t)=\max_{\{i|r_{\ell i}\le t-1\}} g_{\ell i}$.
Accordingly, the average AoI for UE $\ell$ until slot $t$ is $\avA_\ell(t)=\sum_{\tau=1}^t A_\ell(\tau)/t$, and the long-term average AoI
\begin{align} \label{eq:avAoI}
	\avA_\ell=\lim_{t\to\infty}\avA_\ell(t)=\lim_{t\to\infty}\sum_{\tau=1}^t A_\ell(\tau)/t.
\end{align}

\begin{remark} \label{remark:AoI-not-all-packets}
	At any time, the AoI (for a UE) only depends on the arrival time of the latest packet received at the UE. Therefore, for minimizing the AoI for a UE $\ell$, instead of transmitting every packet that arrives for UE $\ell$, the BS may transmit only a subset of packets based on their arrival times. Thus, the remaining slots may be used to serve other UEs. 
\end{remark}

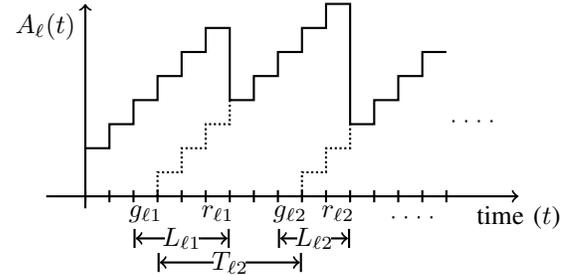
\begin{figure} [htbp] 
	\begin{center}
		\begin{tikzpicture}[thick,scale=0.8, every node/.style={scale=1}]
		\draw[->] (-0.25,0) to (7.6,0) node[below]{time ($t$)};
		\draw[->] (0.4,-0.2) to (0.4,3.2) node[below left]{$A_\ell(t)$};
		\draw (0.4,0.8) to (0.8,0.8) to (0.8,1.2) to (1.2,1.2) to (1.2,1.6) to (1.6,1.6) to (1.6,2) to (2,2) to (2,2.4) to (2.4,2.4) to (2.4,2.8) to (2.8,2.8) to (2.8, 1.6) to (3.2,1.6) to (3.2,2) to (3.6,2) to (3.6,2.4) to (4,2.4) to (4,2.8) to (4.4,2.8) to (4.4,3.2) to (4.8,3.2) to (4.8,1.2) to (5.2,1.2) to (5.2,1.6) to (5.6,1.6) to (5.6,2) to (6,2) to (6,2.4) to (6.4,2.4); 
		
		\draw[densely dotted] (1.6,0) to (1.6,0.4) to (2,0.4) to (2,0.8) to (2.4,0.8) to (2.4,1.2) to (2.8,1.2) to (2.8,1.6);
		\draw (1.4,0) node[below]{$g_{\ell 1}$};
		\draw (2.6,0) node[below]{$r_{\ell 1}$};
		\draw[densely dotted] (4,0) to (4,0.4) to (4.4,0.4) to (4.4,0.8) to (4.8,0.8) to (4.8,1.2);
		\draw (3.8,0) node[below]{$g_{\ell 2}$};
		\draw (4.6,0) node[below]{$r_{\ell 2}$};
		
		\draw[loosely dotted] (6.5,1.25) to (7.3,1.25); 
		
		\draw (0.4,-0.1) to (0.4,0.1);
		\draw (0.8,-0.1) to (0.8,0.1);
		\draw (1.2,-0.1) to (1.2,0.1);
		\draw (1.6,-0.1) to (1.6,0.1);
		\draw (2,-0.1) to (2,0.1);
		\draw (2.4,-0.1) to (2.4,0.1);
		\draw (2.8,-0.1) to (2.8,0.1);
		\draw (3.2,-0.1) to (3.2,0.1);
		\draw (3.6,-0.1) to (3.6,0.1);
		\draw (4,-0.1) to (4,0.1);
		\draw (4.4,-0.1) to (4.4,0.1);
		\draw (4.8,-0.1) to (4.8,0.1);
		\draw (5.2,-0.1) to (5.2,0.1);
		\draw (5.6,-0.1) to (5.6,0.1); 
		\draw (6,-0.1) to (6,0.1);
		\draw (6.4,-0.1) to (6.4,0.1);
		
		\draw[|<->|] (1.2,-0.7) -- (2.8,-0.7) node[rectangle,inner sep=-1pt,midway,fill=white]{$L_{\ell 1}$}; 
		\draw[|<->|] (3.6,-0.7) -- (4.8,-0.7) node[rectangle,inner sep=-1pt,midway,fill=white]{$L_{\ell 2}$};
		\draw[|<->|] (1.6,-1.1) -- (4,-1.1) node[rectangle,inner sep=-1pt,midway,fill=white]{$T_{\ell 2}$};
		
		\draw[loosely dotted] (5.5,-0.35) to (6.3,-0.35);
		
		\end{tikzpicture}
		\caption{Sample plot illustrating AoI $A_\ell(t)$ against time, and latency $L_{\ell i}$ of packet $i$.\vspace{-4ex}} 
		\label{fig:age} 
	\end{center}
\end{figure}

\paragraph{Latency} For any packet, the latency is defined as the number of slots the packet remains at the BS, i.e., the number of slots since its arrival at the BS, to until it is successfully transmitted (received at the destination UE). 
Therefore, if a packet $i$ that arrived at the BS in slot $g_{\ell i}$ is successfully transmitted in slot $r_{\ell i}$, then the latency for the packet is $L_{\ell i}=r_{\ell i}-g_{\ell i}+1$. For UE $\ell$, we define average latency as the average of the latency of all packets that arrive at the BS for UE $\ell$. Therefore, if $\cG_\ell(t)$ denotes the set of packets that arrive at the BS for UE $\ell$ until time $t$, and $|\cG_\ell(t)|$ denotes the cardinality of $\cG_\ell(t)$, then the average latency until time $t$ is equal to 
\begin{align} \label{eq:avLatency-t}
	\avL_\ell(t)=\frac{\sum_{i\in\cG_\ell(t)}L_{\ell i}}{|\cG_\ell(t)|},
\end{align}
where while computing $L_{\ell i}$ for untransmitted packet $i\in\cG_\ell(t)$, we use $r_{\ell i}=t$.

\begin{remark}
Note that when $t\to\infty$, $|\cG_\ell(t)|=q_\ell \cdot t$ with probability $1$ (using the strong law of large numbers). Therefore, the long-term average latency for UE $\ell$,
\begin{align} \label{eq:avLatency}
	\avL_\ell=\lim_{t\to\infty} \avL_\ell(t)
	=\lim_{t\to\infty}\frac{\sum_{i\in\cG_\ell(t)}L_{\ell i}}{q_\ell\cdot t}.
\end{align}
\end{remark}

\begin{remark} \label{remark:latency-tx-every-packet}
	Since $\avL_\ell$ \eqref{eq:avLatency} is a function of the latency of all packets in $\cG_\ell(t)$, for minimizing $\avL_\ell$, the BS must never discard any packet in $\cG_\ell(t)$ (i.e., transmit every packet in $\cG_\ell(t)$), $\forall t$. 
\end{remark}

\paragraph{Throughput}
For a UE, throughput is defined as the average number of packets that it receives from the BS per slot. Formally, if $\cR_\ell(t)$ denotes the set of packets that UE $\ell$ received from the BS until time $t$, then the throughput of UE $\ell$ is given as 
\begin{align} \label{eq:throughput}
	\avR_\ell=\lim_{t\to\infty}\frac{|\cR_\ell(t)|}{t},
\end{align} 
where $|\cR_\ell(t)|$ denotes the cardinality of set $\cR_\ell(t)$.

Based on the applications/QoS requirements of the UEs, we partition the set of UEs  $\cI$ into subsets $\cI_A$, $\cI_L$ and $\cI_R$. i) For each UE $i\in\cI_A$, the average AoI needs to be `small', ii) for each UE $j\in\cI_L$, the average latency needs to be `small', and iii) for each UE $k\in\cI_R$, the throughput should be greater than a fixed threshold (that may be different for each UE). 

\begin{remark} \label{remark:q=1}
	Note that in any slot, if the BS does not have any packet for UE $k\in\cI_R$ (i.e., the BS has transmitted all the packets that arrived for UE $k$ so far), then the UE $k$ already has maximum possible throughput until the current slot (best case for UE $k$).  
	Thus, without loss of generality, we assume that the BS always has some packet for each UE $k\in\cI_R$ to transmit (or equivalently, let the arrival rate $q_k=1$, $\forall k\in\cI_R$). 
\end{remark}

\begin{definition} \label{def:causal-policy}
	A centralized causal scheduling policy (in short, causal policy) is an algorithm that in each slot $t$, using only the causal information available at the BS, decides which packet/UE the BS must transmit/serve.
\end{definition}

In light of the QoS requirements of UEs in $\cI$, the goal is to find a causal policy (Definition \ref{def:causal-policy}) that simultaneously minimizes the average AoI of the UEs in $\cI_A$, the average latency of the UEs in $\cI_L$, and satisfies the throughput constraints of the UEs in $\cI_R$.  
We formulate this objective as the following optimization problem.
\begin{subequations} \label{eq:prob-form-2}
	\begin{align} \label{eq:alt-objective}
		\min_{\pi\in\Pi} \ \ & \sum_{i\in\cI_A}\rho_i\avA_i^\pi, \\ 
		\label{eq:alt-constraint-latency}
		s.t. \ \ & \avL_j^\pi\le \beta_j, \ \ \forall j\in\cI_L, \\
		\label{eq:alt-constraint-throughput}
		\ \ & \avR_k^\pi\ge \alpha_k, \ \ \forall k\in\cI_R,
	\end{align}
\end{subequations}
where $\Pi$ denotes the set of all causal policies, $\avA_i^\pi, \avL_j^\pi$ and $\avR_k^\pi$ respectively denotes the long-term average AoI \eqref{eq:avAoI}, long-term average latency \eqref{eq:avLatency} and average throughput \eqref{eq:throughput} under policy $\pi$,  and $\rho_i,\beta_j,\alpha_k>0$ are constant parameters, $\forall i\in\cI_A$, $\forall j\in\cI_L$ and $\forall k\in\cI_R$.

In the rest of this paper, we solve problem \eqref{eq:prob-form-2} as follows. First, in Section \ref{sec:preliminaries} we consider a related problem where the objective is to minimize a weighted sum of the average AoI of UEs in $\cI_A$ and the average latency of UEs in $\cI_L$, and propose a candidate policy $\pi_H$ to solve that problem. Subsequently, in Section \ref{sec:find-Lagrangian} we generalize $\pi_H$ and design two policies $\pi_{VW}$ and $\pi_{RD}$ to solve problem \eqref{eq:prob-form-2}. Finally, in Section \ref{sec:numerical-results} we analyze the performance of policy $\pi_H$ and its variants using numerical simulations.

\section{A Related Problem} \label{sec:preliminaries}

Note that both AoI and latency are time-sensitive, i.e. they depend on the slot in which a packet is transmitted. However, they are difficult to compare in problem \eqref{eq:prob-form-2}, where the AoI terms appear in the objective, whereas the latency is in the constraint. Therefore, consider the following alternate formulation of problem \eqref{eq:prob-form-2}, with both the AoI and latency in the objective function:
\begin{subequations} \label{eq:prob-form-1}
	\begin{align} \label{eq:objective}
		\min_{\pi\in\Pi} \ \ & \sum_{i\in\cI_A}\rho_i\avA_i^\pi+\sum_{j\in\cI_L} \rho_j \avL_j^\pi, \\
		\label{eq:constraint}
		s.t. \ \ & \avR_k^\pi\ge \alpha_k, \ \ \forall k\in\cI_R,
	\end{align}
\end{subequations}
where $\rho_j\ge 0$ ($\forall j\in\cI_L$) are known constants. The intuition is that if we can solve problem \eqref{eq:prob-form-1} for any $\rho_j\ge 0$, then we can solve problem \eqref{eq:prob-form-2} by finding suitable $\rho_j$'s, and converting \eqref{eq:prob-form-2} into \eqref{eq:prob-form-1}. Therefore, first we try to solve problem \eqref{eq:prob-form-1}. For this, we need some preliminaries that we discuss next.

\subsection{Preliminaries}

Let $\Pi_F\subseteq \Pi$ denote the set of all feasible policies, i.e., the set of causal policies for which the cost \eqref{eq:objective} is finite, and constraint \eqref{eq:constraint} is satisfied $\forall k\in\cI_R$. 

\begin{lemma} \label{lemma:throughput-UEs}
	For any policy $\pi\in\Pi_F$, the throughput
	\begin{enumerate}
		\item $\avR_i^\pi>0$, $\forall i\in\cI_A$,
		\item $\avR_j^\pi=q_j$, $\forall j\in\cI_L$,
		\item $\avR_k^\pi\ge\alpha_k$,  $\forall k\in\cI_R$, and
		\item $\sum_{i\in\cI_A}\avR_i^\pi/p_i\le 1-\sum_{j\in\cI_L}q_j/p_j-\sum_{k\in\cI_R}\alpha_k/p_k$.
	\end{enumerate}
\end{lemma}
\begin{proof}
	See Appendix \ref{app:proof-lemma-throughput-UEs}.
\end{proof}

\begin{theorem} \label{thm:feasibility}
	The optimization problem \eqref{eq:prob-form-1} is feasible, i.e., the set $\Pi_F$ is non-empty, only if $\sum_{j\in\cI_L}\frac{q_j}{p_j}+\sum_{k\in\cI_R}\frac{\alpha_k}{p_k}<1$.
\end{theorem}
\begin{proof}
	For $\Pi_F$ to be non-empty, there must exist some policy $\pi\in\Pi$ that satisfies all the throughput relations of Lemma \ref{lemma:throughput-UEs}. However, any policy $\pi$ may satisfy throughput relations $1$ and $4$ (of Lemma \ref{lemma:throughput-UEs}) only if $\sum_{j\in\cI_L}q_j+\sum_{k\in\cI_R}\alpha_k<1$. Hence, $\sum_{j\in\cI_L}q_j+\sum_{k\in\cI_R}\alpha_k<1$ is a necessary condition for $\Pi_F$ to be non-empty.  
\end{proof}

\begin{remark} \label{remark:feasible-problem}
In the rest of this paper, without loss of generality, we assume that problem \eqref{eq:prob-form-1} is feasible, i.e. $\sum_{j\in\cI_L}q_j/p_j+\sum_{k\in\cI_R}\alpha_k/p_k<1$.
\end{remark}

Next, for policies $\pi\in\Pi_F$, we derive an equivalent formulation for the optimization problem \eqref{eq:prob-form-1}. To this end, we define some notations as follows. 
Under policy $\pi\in\Pi_F$, let $\cR_\ell^\pi(t)$ denote the subset of packets that the BS successfully transmits to UE $\ell\in\cI$, until slot $t$. 
Also, let $g_{\ell i}^\pi$ and $r_{\ell i}^\pi$ respectively denote the slots in which the $i^{th}$ packet of $\cR_{\ell i}^\pi(t)$ ($i^{th}$, in increasing order of arrival time at the BS) arrived at the BS, and got received at UE $\ell$. 
By definition, under policy $\pi\in\Pi_F$, the latency $L_{\ell i}^\pi=r_{\ell i}^\pi-g_{\ell i}^\pi+1$. 

Further, define $T_{\ell i}^\pi=g_{\ell i}^\pi-g_{\ell (i-1)}^\pi$ as the inter-arrival time of packets in $\cR_\ell^\pi(t)$ (see Figure \ref{fig:age}). 
Also, let
	\begin{align} \label{eq:avT}
		\avT_\ell^\pi=\lim_{t\to\infty}\frac{\sum_{i\in\cR_\ell^\pi(t)}T_{\ell i}^\pi}{|\cR_\ell^\pi(t)|} 
	\end{align} 
	denote the average inter-arrival time of packets in $\cR_\ell^\pi(t)$ (i.e., the average inter-arrival time of packets that are received at UE $\ell\in\cI$), and $\delta_{\ell i}^\pi=T_{\ell i}^\pi-\avT_\ell^\pi$. 
\begin{lemma} \label{lemma:relation-throughput-cycle-length}
	For each UE $\ell\in\cI$, $\avT_\ell^\pi=1/\avR_\ell^\pi$.
\end{lemma}
\begin{proof}
	See Appendix \ref{app:proof-lemma-relation-throughput-cycle-length}.
\end{proof}

The average AoI of a UE can be expressed as follows.
\begin{lemma} \label{lemma:avAoI-general-expression}
	For any policy $\pi\in\Pi_F$ and UE $\ell\in\cI_A$, the average AoI $\avA_\ell^\pi=$
	\begin{align} \label{eq:avAoI-equiv}
		\frac{1}{2}\bigg[\avT_\ell^\pi+\frac{\avdelta}{\avT_\ell^\pi}+1\bigg]+\lim_{t\to\infty}\sum_{i\in\cR_\ell^\pi(t)}\frac{T_{\ell i}^\pi (L_{\ell i}^\pi-1)}{t} 
	\end{align}
	 with probability $1$, where 
	$$ \avdelta=\lim_{t\to\infty}\sum_{i\in\cR_\ell^\pi(t)}\frac{(\delta_{\ell i}^\pi)^2}{|\cR_\ell^\pi(t)|}=\lim_{t\to\infty}\sum_{i\in\cR_\ell^\pi(t)}\frac{(T_{\ell i}^\pi-\avT_\ell^\pi)^2}{|\cR_\ell^\pi(t)|}. 
	$$ 
\end{lemma} 
\begin{proof}
	See Appendix \ref{app:proof-lemma-avAoI-general-expression}.
\end{proof}

Substituting \eqref{eq:avAoI-equiv} and \eqref{eq:avLatency} into \eqref{eq:objective}, and rearranging terms, we get the following equivalent formulation for problem \eqref{eq:prob-form-1}:
\begin{align} \label{eq:objective-equiv}
	\min_{\pi\in\Pi_F} \ \ & F_1^\pi+F_2^\pi, 
\end{align}
where
\begin{align} \label{eq:objective-equiv-term1}
F_1^\pi&=\sum_{\ell\in\cI_A}\frac{\rho_\ell}{2}\bigg[\avT_\ell^\pi+\frac{\avdelta}{\avT_\ell^\pi}+1\bigg], \\ 
\label{eq:objective-equiv-term2}
F_2^\pi&=\lim_{t\to\infty} \bigg[\sum_{\stackrel{\ell\in\cI_A}{i\in\cR_\ell^\pi(t)}}\frac{(\rho_\ell T_{\ell i}^\pi) (L_{\ell i}^\pi-1)}{t}+\sum_{\stackrel{\ell\in\cI_L}{i\in\cG_\ell(t)}} \frac{(\rho_\ell/q_\ell)L_{\ell i}^\pi}{t}\bigg], 
\end{align}
and the throughput constraint \eqref{eq:constraint} is accounted in the fact that the optimization is only over policies $\pi\in\Pi_F$ (that by definition, satisfy \eqref{eq:constraint}).

\begin{lemma} \label{lemma:avdelta-lb}
	Under an optimal policy $\pi\in\Pi_F$ (policy that minimizes \eqref{eq:objective-equiv}), $\avdelta > (1-q_\ell)/2q_\ell^2$ for each UE $\ell\in\cI_A$. Therefore, from \eqref{eq:objective-equiv-term1}, we get
	\begin{align} \label{eq:F1-lb}
		F_1^\pi> \sum_{\ell\in\cI_A}\frac{\rho_\ell}{2}\bigg[\avT_\ell^\pi+\frac{(1-q_\ell)/2q_\ell^2}{\avT_\ell^\pi}+1\bigg].
	\end{align}
\end{lemma}
\begin{proof}
	See Appendix \ref{app:proof-avdelta-lb}.
\end{proof}

Next, we propose a causal policy as a candidate solution to the optimization problem \eqref{eq:objective-equiv}.

\subsection{Algorithm} \label{sec:algorithm}

\begin{definition} \label{def:notations-algo}
	In all the algorithms (pseudocodes) in this paper,
	$\hat{R}_\ell(t)$ and $|\cR_\ell(t)|$ respectively denotes the number of transmission attempts, and the number of successful transmissions to UE $\ell$, until slot $t$. Also, $\lambda_\ell(t)=\max_{\{i|r_{\ell i}\le t-1\}} g_{\ell i}$ denotes the arrival time of the latest packet at the BS that has been successfully transmitted to UE $\ell$ by the end of slot $t-1$.
\end{definition}

Consider the hierarchical index policy $\pi_H$ (Algorithm \ref{algo:policy}) that primarily follows two steps.
\begin{algorithm}
	\caption{Hierarchical Index Policy $\pi_H$.}
	\label{algo:policy}
	\begin{algorithmic} [1]
		\STATE Initialize variables $a_\ell=0$, $b_\ell=0$ ($\forall \ell\in\cI_A$), and $W_\ell=0$ ($\forall \ell\in\cI_A\cup\cI_L$); 
		\STATE Compute $\avT_\ell^\star$ \eqref{eq:compute-Tstar};
		\FOR{each slot $t\in\{1,2,3,\cdots\}$} 
		\STATE \textbf{Update\_Index\_UEs\_$\mathbf{\cI_A\cup\cI_L}$} (Codeblock \ref{algo:policy-Packet-Prioritization}); 		
		\STATE \textbf{Transmit\_Packet} (Codeblock \ref{algo:policy-Packet-Transmission});
		\ENDFOR
	\end{algorithmic}
\end{algorithm}

\paragraph{Index Assignment to UEs $\ell\in\cI_A\cup\cI_L$} 
$\pi_H$ maintains a set $\cS$ of high priority untransmitted packets for UEs $\ell\in\cI_A\cup\cI_L$, and each of these packets is assigned an index at arrival. In any slot, for UE $\ell$, the index $W_\ell\in\cI_A\cup\cI_L$ is defined to be equal to the index of its latest packet in $\cS$. If $\cS$ does not have any packet for UE $\ell$, then $W_\ell=0$.

\begin{figure}[htbp]
	\centerline{\includegraphics[scale=0.4]{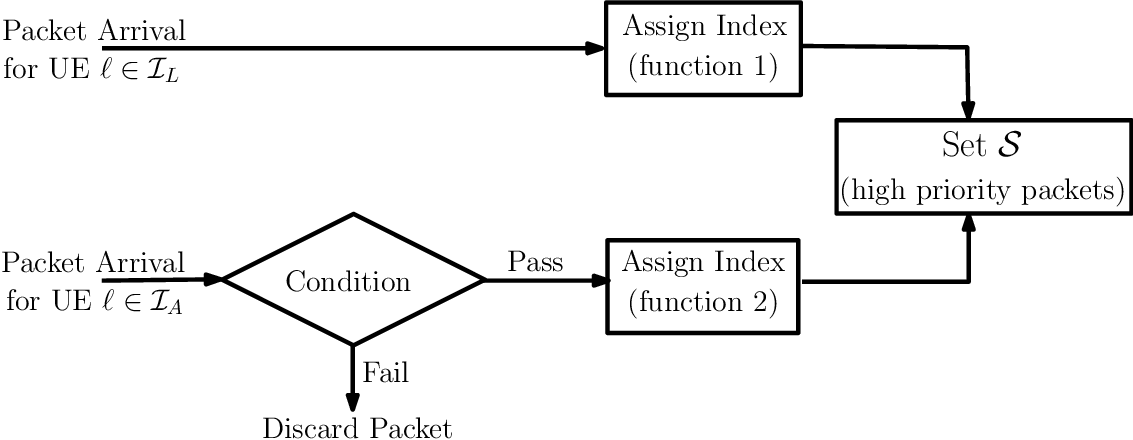}}
	\caption{[Flowchart] $\pi_H$ selects high priority packets.}
	\label{fig:pkt_selection}
\end{figure}

The procedure for adding packets to set $\cS$ (Figure \ref{fig:pkt_selection}) is as follows. In any slot $t$, if a packet arrives for UE $\ell\in\cI_L$, then the packet is immediately added to $\cS$, and assigned an index $\rho_\ell p_\ell/q_\ell$ (because all packets for UE $\ell\in\cI_L$ are equally important, and must be transmitted; Remark \ref{remark:latency-tx-every-packet}). If a packet arrives for UE $\ell\in\cI_A$ (in slot $t$), then the packet is added to $\cS$ if the throughput for UE $\ell$ until slot $t$ is less than $1/\avT_\ell^\star$, where $\avT_\ell^\star$ is obtained by solving the following convex optimization problem \eqref{eq:compute-Tstar}, 
\begin{subequations} \label{eq:compute-Tstar}
\begin{align} \label{eq:threshold-compute}
	&\underset{\avT_\ell,\forall \ell\in\cI_A}{\arg\min} \ \  \sum_{\ell\in\cI_A} \frac{\rho_\ell}{2}\left(\avT_\ell+\frac{(1-q_\ell)/q_\ell^2}{\avT_\ell}\right),  \\	\label{eq:avT-constraint}
	&\text{ s.t.} \ \ \sum_{\ell\in\cI_A}\frac{1}{p_\ell\avT_\ell}\le \zeta, \textbf{ and } \avT_\ell\ge 1 \ \ (\forall \ell\in\cI_A), 
\end{align}
\end{subequations}
where $\zeta=1-\sum_{j\in\cI_L}q_j/p_j-\sum_{k\in\cI_R}\alpha_k/p_k$ is a constant. 

\begin{remark}
	The function being minimized in \eqref{eq:threshold-compute} (with some additive constants) is an upper bound on $F_1^H$ \eqref{eq:objective-equiv-term1} (where superscript $H$ denotes policy $\pi_H$) in terms of average inter-arrival time of transmitted packets $\avT_\ell$'s.
	Since $\avT_\ell=1/\avR_\ell$ (Lemma \ref{lemma:relation-throughput-cycle-length}), problem \eqref{eq:compute-Tstar} can be interpreted as minimizing the upper bound on $F_1^H$ with respect to $\avT_\ell$'s, under constraint $\sum_{\ell\in\cI_A}\avR_\ell/p_\ell \le 1-\sum_{j\in\cI_L}q_j/p_j-\sum_{k\in\cI_R}\alpha_k/p_k$ derived in Lemma \ref{lemma:throughput-UEs}. $\avT_\ell\ge 1$ because inter-arrival time of packets (at the BS) for any UE is at least one slot. 
\end{remark}

For technical reasons (discussed later), for adding packets of UE $\ell\in\cI_A$ to set $\cS$, instead of directly comparing the throughput of UE $\ell$ with $1/\avT_\ell^\star$, $\pi_H$ uses a threshold-based counter $b_\ell$. Let $a_\ell(t)$ denote the latest slot (until slot $t-1$) in which $b_\ell$ got incremented. 
In slot $t$, if a packet arrives for UE $\ell$, and $t-a_\ell(t)>\lceil \avT_\ell^\star-1/q_\ell\rceil$, then $b_\ell$ is incremented by $1$. Further, in any slot $t$, if $b_\ell>|\cR_\ell(t)|$ (the number of packets successfully transmitted to UE $\ell$ until slot $t$), and a packet $i$ arrives for UE $\ell$, then the packet is added to set $\cS$, and assigned an index $\rho_\ell T_{\ell i} p_\ell=\rho_\ell(t-\lambda_\ell(t))p_\ell$. 
For UE $\ell\in\cI_A$, since AoI only depends on the latest transmitted packet, and the reduction in AoI is largest on transmission of latest packet, whenever a new packet for UE $\ell$ is added to set $\cS$, all its previous packets in $\cS$ are discarded. 

\begin{algorithm}
	\floatname{algorithm}{Codeblock}
	\caption{\textbf{Update\_Index\_UEs\_$\mathbf{\cI_A\cup\cI_L}$}}
	\label{algo:policy-Packet-Prioritization}
	\begin{algorithmic} [1]
		\IF{a packet arrives for UE $\ell\in\cI_A\cup\cI_L$}
		\IF{$\ell\in\cI_L$}
		\STATE $W_\ell=\rho_\ell p_\ell/q_\ell$;
		\ELSE 
		\IF{$t-a_\ell\ge \lceil \avT_\ell^\star-1/q_\ell\rceil$}
		\STATE $a_\ell=t$, $b_\ell=b_\ell+1$;
		\ENDIF
		\IF{$b_\ell>|\cR_\ell(t-1)|$}
		\STATE $W_\ell=\rho_\ell p_\ell (t-\lambda_\ell(t))$;
		\ENDIF
		\ENDIF
		\ENDIF
	\end{algorithmic}
\end{algorithm}

\begin{figure}[htbp]
	\centerline{\includegraphics[scale=0.5]{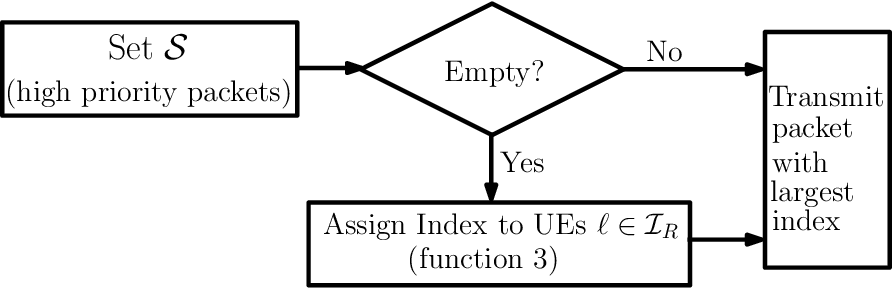}}
	\caption{[Flowchart] $\pi_H$ selects packet to transmit.}
	\label{fig:pkt_transmission}
\end{figure}
\paragraph{Packet transmission}

As shown in Figure \ref{fig:pkt_transmission}, in each slot $t$, if set $\cS$ is non-empty ($W_\ell>0$ for some UE $\ell\in\cI_A\cup\cI_L$), then $\pi_H$ transmits the highest index packet from $\cS$ (i.e., the latest packet of UE $\arg\max_{\ell\in\cI_A\cup\cI_L}W_\ell$), and removes the packet from $\cS$ if the transmission is successful. 
If $\cS$ is empty (i.e. $W_\ell=0$, $\forall \ell\in\cI_A\cup\cI_L$), then to each UE $k\in\cI_R$, $\pi_H$ assigns index $(\alpha_k t/p_k)-\hat{R}_k(t)$ (where $\hat{R}_k(t)$ denotes the number of slots until slot $t$, in which BS makes a transmission attempt to UE $k$), and transmits packet of UE $k^\star\in\cI_R$ for which the index is maximum.

\begin{remark} 
	Under $\pi_H$, indices $W_\ell$ of UEs in $\ell\in\cI_A\cup\cI_L$ are never directly compared to the indices of UEs in $\cI_R$. There is a clear hierarchy in the priority of the two sets of UEs. Therefore, we call $\pi_H$, the `\emph{hierarchical index policy}'. 
\end{remark}

\begin{algorithm}
	\floatname{algorithm}{Codeblock}
	\caption{\textbf{Transmit\_Packet}}
	\label{algo:policy-Packet-Transmission}
	\begin{algorithmic} [1]
		\STATE $\ell^\star=\arg\max_{\ell\in\cI_A\cup\cI_L} W_\ell$;
		\IF {$W_{\ell^\star}>0$}
		\STATE transmit the latest packet of UE $\ell^\star$;
		\IF{the transmission is successful}
		\IF{$\ell^\star\in\cI_A$}
		\STATE $W_{\ell^\star}=0$; 
		\ELSIF{no remaining packet for UE $\ell^\star\in\cI_L$}
		\STATE $W_{\ell^\star}=0$;
		\ENDIF
		\ENDIF
		\ELSE
		\STATE transmit packet of UE $k^\star=\underset{k\in\cI_R}{\arg\max}\left\{\frac{\alpha_k\cdot t}{p_k}-\hat{R}_k(t)\right\}$;
		\ENDIF
	\end{algorithmic}
\end{algorithm}

The basic idea behind policy $\pi_H$ is as follows. As shown in Lemma \ref{lemma:throughput-UEs}, for problem \eqref{eq:objective-equiv}, each UE $\ell\in\cI_L\cup\cI_R$ needs to have a certain throughput, that requires the BS to transmit to each of these UEs in certain proportion of the total number of slots. Therefore, for UEs $\ell\in\cI_A$, a policy must choose a subset $\cS_A$ of arriving packets that can be transmitted in the remaining slots (i.e., the throughput of UEs in $\cI_A$ satisfy Relation $4$ of Lemma \ref{lemma:throughput-UEs}). $\pi_H$ chooses the subset $\cS_A$ (and adds it to set $\cS$) using a threshold-based counter $b_\ell$, which ensures that the packets are evenly spread out in time such that with the limited number of transmissions, maximum reduction in average AoI can be achieved. 
The rationale for the choice of threshold $\lceil\avT_\ell^\star-1/q_\ell\rceil$ (for incrementing $b_\ell$) is discussed next. 

\begin{lemma} \label{lemma:b-increment-packet-addition}
	In any slot in which $b_\ell$ is incremented, a new packet for UE $\ell$ is added to set $\cS$. 
\end{lemma}
\begin{proof}
	Note that $|\cR_\ell(t)|$ can never be greater than $b_\ell$ (because for UE $\ell\in\cI_A$, only packets in set $\cS$ may be transmitted, and for UE $\ell\in\cI_A$, packets are added to set $\cS$ only when $b_\ell>|\cR_\ell(t)|$). Hence, in the slot in which $b_\ell$ is incremented, $b_\ell>|\cR_\ell(t)|$. Also, $b_\ell$ is incremented only if a new packet arrives in the slot for UE $\ell$. Thus, when $b_\ell$ is incremented, a packet for UE $\ell$ is added to set $\cS$, in the same slot.
\end{proof}

As shown in Lemma \ref{lemma:b-increment-packet-addition}, when $b_\ell$ is incremented, a new packet for UE $\ell$ is also added to set $\cS$.
Since $\cS$ is the set of high priority packets, it is expected that the packet will soon be transmitted (i.e., most likely the packet will not be replaced by newer packet). Thus, assuming the sequence of packets for UE $\ell\in\cI_A$ that are added to set $\cS$ immediately following increment of $b_\ell$, are all transmitted, if the threshold $\lceil\avT_\ell^\star-1/q_\ell\rceil$ (for incrementing $b_\ell$) is used, $\pi_H$ will have the following key properties. 
\begin{remark}
	In all the notations in this paper, superscript $H$ represents policy $\pi_H$.
\end{remark}
\begin{theorem} \label{thm:F1H-ub}
	For policy $\pi_H$, $\avT_\ell^H\in[\avT_\ell^\star, \avT_\ell^\star+1)$, and  
		$$\overline{(\delta_\ell^H)^2}=\lim_{t\to\infty}\sum_{i\in\cR_\ell^H(t)}\frac{(T_{\ell i}^H-\avT_\ell^H)^2}{|\cR_\ell^H(t)|}= \frac{(1-q_\ell)}{q_\ell^2},$$
		$\forall \ell\in\cI_A$. 
		Thus, using \eqref{eq:objective-equiv-term1}, we get 
		\begin{align} \label{eq:F1H-ub}
			F_1^H &< \sum_{\ell\in\cI_A}\frac{\rho_\ell}{2}\bigg[(\avT_\ell^\star+1)+\frac{(1-q_\ell)/q_\ell^2}{\avT_\ell^\star}+1\bigg], \\
			\label{eq:F1H-guarantee}
			&< 2\cdot\underset{\pi\in\Pi_F}{\min} F_1^\pi.
		\end{align}
	\end{theorem}
	\begin{proof}
		See Appendix \ref{app:proof-thm-F1H-ub}.
	\end{proof}

Theorem \ref{thm:F1H-ub} shows that when the threshold $\lceil\avT_\ell^\star-1/q_\ell\rceil$ is used, $F_1^H$ \eqref{eq:objective-equiv-term1} is close to optimal. Therefore, for minimizing the objective \eqref{eq:objective-equiv}, it only remains to minimize $F_2^H$ \eqref{eq:objective-equiv-term2}, that $\pi_H$ does in the packet transmission step.

In `packet transmission' step, packets in $\cS$ are given priority over all UEs $k\in\cI_R$. This is because unlike the packets for UEs $k\in\cI_R$, the packets in $\cS$ are time-sensitive (longer the packets remain at the UEs, larger is the cost).
Note that $\pi_H$ giving more priority to packets in $\cS$ (i.e. packets for UEs $\cI_A\cup\cI_L$) does not ignore UEs $k\in\cI_R$. This is because the threshold-based counter $b_\ell$ restricts the rate at which packets are added to set $\cS$, thus ensuring that sufficient slots are available for UEs $k\in\cI_R$ to transmit, and satisfy the constraint \eqref{eq:constraint} (Theorem \ref{thm:throughput-constraint-satisfied-piH}).

Next, we discuss the rationale for the choice of indices assigned to the packets in $\cS$ and UEs in $\cI_R$, and how they help in minimizing $F_2^H$ \eqref{eq:objective-equiv-term2}. Note that $F_2^\pi$ \eqref{eq:objective-equiv-term2} is a weighted average of the sum of latency of packets that the BS transmits to UEs $\ell\in\cI_A\cup\cI_L$ (includes all packets that arrive for UEs $\ell\in\cI_L$). From prior works \cite{buyukkoc1985cmu,krishnasamy2018learning}, it is well-known that for minimizing weighted sum of latency of packets per time slot, the $c\mu$-rule, i.e., in each slot, transmit the packet for which the product of the weight and the transmission success probability is maximum, is optimal. Therefore, for each packet in $\cS$, $\pi_H$ assigns an index that is equal to the product of the weight of the packet and the corresponding transmission success probability, and in each slot, transmits the packet with the largest index.
\begin{remark}
	Note that in set $\cS$, the weight $\rho_\ell T_{\ell i}$ (index $\rho_\ell T_{\ell i}p_\ell$) of packet $i$ for UE $\ell\in\cI_A$ changes if the packet is replaced by a newly arrived packet.  
	Therefore, the considered setting is different from the setting for which $c\mu$-rule is known to be optimal. However, $c\mu$-rule is still a promising candidate, and showing it optimal/near-optimal for the considered setting remains a part of our ongoing work. 
\end{remark}

Further, in \cite{hou2009theory}, it has been shown that when there are multiple UEs, a policy that in each slot, transmits to UE $k$ for which $(\alpha_k t/p_k)-\hat{R}_k(t)$ (i.e., the expected number of transmission attempts $\alpha_k t/p_k$ needed to satisfy the throughput constraint, and the difference between the actual number of transmission attempts $\hat{R}_k(t)$) is maximum, satisfies the throughput constraint of all the UEs (when the constraints are feasible; Theorem \ref{thm:feasibility}). Therefore, $\pi_H$ assigns index $Y_k=(\alpha_k t/p_k)-\hat{R}_k(t)$ to UEs $k\in\cI_R$, and whenever set $\cS$ is empty, transmits the packet for UE $k^\star=\arg\max_k Y_k$.

\begin{theorem} \label{thm:throughput-constraint-satisfied-piH}
	$\pi_H$ satisfies the throughput constraint \eqref{eq:constraint}, for all UEs $k\in\cI_R$.
\end{theorem}
\begin{proof}
	Recall that in each of the $\sum_{k\in\cI_k}\hat{R}_k^H(t)$ many slots in which $\pi_H$ transmits to UEs in $\cI_R$, $\pi_H$ transmits to the UE $k\in\cI_R$ for which $t\alpha_k/p_k-\hat{R}_k^H(t)$ is largest. In \cite{hou2009theory} (Theorem $2$), it has been shown that if $\underset{t\to\infty}{\lim}\sum_{k\in\cI_R}\hat{R}_k^H(t)/t\ge \sum_{k\in\cI_R}\alpha_k/p_k$, then the above step is feasibility optimal, i.e., $\avR_k\ge\alpha_k$, $\forall k$ with probability $1$. Also, we show in Appendix \ref{app:proof-lemma-throughput-piH} that under $\pi_H$, $\underset{t\to\infty}{\lim}\sum_{k\in\cI_R}\hat{R}_k^H(t)/t\ge \sum_{k\in\cI_R}\alpha_k/p_k$. Hence, $\pi_H$ satisfies the throughput constraint \eqref{eq:constraint} for all UEs $k\in\cI_R$.
\end{proof}

\section{Algorithms for Problem $\eqref{eq:prob-form-2}$} \label{sec:find-Lagrangian}

In this section, we propose two variants of policy $\pi_H$ (Algorithm \ref{algo:policy}) as candidate solutions of problem \eqref{eq:prob-form-2}: i) hierarchical index policy with virtual weights $\pi_{VW}$, and ii) hierarchical index policy with randomization $\pi_{RD}$. 

\subsubsection{Hierarchical Index Policy with Virtual Weights $\pi_{VW}$}

Consider policy $\pi_{VW}$ (Algorithm \ref{algo:find-weights}), that to begin with, assigns some (virtual) weight $\rho_j$ to each UE $j\in\cI_L$. Subsequently, in each time slot $t$, $\pi_{VW}$ follows the hierarchical index policy $\pi_H$ (Algorithm \ref{algo:policy}) to minimize the cost function \eqref{eq:objective-equiv}, and periodically (after every $f$ number of slots, for some constant $f\in\bbN$), updates the virtual weight $\rho_j$, for each UE $j\in\cI_L$, depending on the difference between the current average latency $\avL_j(t)$ \eqref{eq:avLatency-t} and the maximum tolerable average latency $\beta_j$ \eqref{eq:alt-constraint-latency}.
In particular, for updating $\rho_j$'s, $\pi_{VW}$ uses gradient-decent based technique, whereby the updated weight $\rho_j:=\max\{0,\rho_j-\eta \cdot (\beta_j-\avL_j(t))\}$ for some constant step size $\eta>0$. 

Note that in $\pi_H$, the role of $\rho_j$'s (specifically for UEs $j\in\cI_L$) is only in deciding which packet to transmit from set $\cS$. If $\rho_j$ is large, then the corresponding index $\rho_j p_j/q_j$ will also be large. Thus, the latest packet of UE $j$ will be given higher priority compared to the other packets in set $\cS$. Therefore, the average latency of UE $j$ will be smaller. Similarly, the average latency of UE $j$ will be large if $\rho_j$ is small. If $\rho_j$ is small relative to the weights $\rho_i$ for UEs $i\in\cI_A$, then packets for UEs $i\in\cI_A$ will be given higher priority, thus minimizing their average latency (average AoI of UE $i$).
While updating $\rho_j$'s under $\pi_{VW}$, the idea is to converge to the smallest weight $\rho_j\ge 0$ for which the latency constraint \eqref{eq:alt-constraint-latency} is satisfied (average latency $\avL_j$ \eqref{eq:avLatency} is less than $\beta_j$). 

\begin{theorem}
	$\pi_{VW}$ satisfies the throughput constraint \eqref{eq:alt-constraint-throughput} (if the constraint is feasible; Theorem \ref{thm:feasibility}).
\end{theorem}
\begin{proof}
	Note that the only difference in $\pi_{VW}$ and $\pi_H$ is that the weights $\rho_j$, for $j\in\cI_L$ are variable. Since $\rho_j$'s, for $j\in\cI_L$, have no role in deciding the size of set $\cS$ (i.e. the number of high priority packets), they cannot decrease the number of slots in which set $\cS$ is empty (i.e., the number of slots in which UEs $k\in\cI_R$ get to transmit). Hence, the throughput of UEs $k\in\cI_R$ under $\pi_{VW}$ is identical to the  throughput under $\pi_H$. Since $\pi_H$ satisfies the throughput constraint \eqref{eq:alt-constraint-throughput} (Theorem \ref{thm:throughput-constraint-satisfied-piH}), this implies that $\pi_{VW}$ must also satisfy the throughput constraint \eqref{eq:alt-constraint-throughput}.
\end{proof}
 
\begin{algorithm}
	\caption{Hierarchical Index Policy with Virtual Weights.} 
	\label{algo:find-weights}
	\begin{algorithmic} [1]
		\STATE Define constants: $f\in\bbN$ and $\eta>0$;
		\STATE Initialize variables $a_\ell=0$, $b_\ell=0$ ($\forall \ell\in\cI_A$) and $W_\ell=0$, ($\forall \ell\in\cI_A\cup\cI_L$); 
		\STATE Initialize (virtual) weight $\rho_j=1$, $\forall j\in\cI_L$;
		\STATE Compute $\avT_\ell^\star$ \eqref{eq:compute-Tstar};
		\FOR{each slot $t\in\{1,2,3,\cdots\}$}
		\STATE \emph{/* Update Virtual Weights $\rho_j$ ($\forall j\in\cI_L$) */}  
		\IF{$t\pmod {f} == 0$}
		\STATE update $\rho_j:=\max\{0,\rho_j-\eta \cdot (\beta_j-\avL_j(t))\}$, $\forall j\in\cI_L$;
		\ENDIF
		\STATE \textbf{Update\_Index\_UEs\_$\mathbf{\cI_A\cup\cI_L}$} (Codeblock \ref{algo:policy-Packet-Prioritization}); 		
		\STATE \textbf{Transmit\_Packet} (Codeblock \ref{algo:policy-Packet-Transmission});
		\ENDFOR
	\end{algorithmic}
\end{algorithm}

\subsubsection{Hierarchical Index Policy with Randomization $\pi_{RD}$}

Consider a Geo/Geo/1 queue with infinite buffer space (and single server), where packet inter-arrival and service times are geometrically distributed with rate $q_\ell$ and $p_\ell$ respectively. 
As shown in \cite{srikant2013communication} (sec. 3.4.2), for such a queue, under all work-conserving scheduling algorithms\footnote{Scheduling algorithms that always serve some packet unless the queue is empty.} the average latency \eqref{eq:avLatency} is $\frac{1-p_\ell}{p_\ell-q_\ell}+1$.\footnote{The additive constant $1$ follows because as per the considered definition, minimum latency for any packet is at least $1$.} 
Therefore, if $p_\ell\ge q_\ell+(1-q_\ell)/\beta_\ell$, then the average latency will be at most $\beta_\ell$.
Next, we use this result to derive a policy that under the following assumption, guarantees to satisfy both the latency and throughput constraints \eqref{eq:alt-constraint-latency}--\eqref{eq:alt-constraint-throughput}. 

\begin{assumption} \label{assume:sum-theta-less-than-1}
	For each UE $j\in\cI_L$, let 
	\begin{align} \label{eq:theta-j}
		\theta_j=\frac{1}{p_\ell}\left(q_j+\frac{1-q_j}{\beta_j}\right).
	\end{align}
For the rest of this section, we assume that $\sum_{j\in\cI_L}\theta_j\le 1$.
\end{assumption}

Consider the `hierarchical index policy with randomization' $\pi_{RD}$ (Algorithm \ref{algo:alt-policy}) that operates in $3$ steps. In step $1$,  
$\pi_{RD}$ selects high priority packets for UEs $i\in\cI_A$ using the threshold-based counter $b_\ell$, and assigns them an index exactly as $\pi_H$ does. However, unlike $\pi_H$, UEs $j\in\cI_L$ are not considered in this step. 
In step $2$, $\pi_{RD}$ selects the UE $\ell^\star\in\cI_A$ with largest index $W_{\ell^\star}$. If for all UEs in $I_A$ the index is $0$ (i.e., no packet to transmit to UEs in $\cI_A$), (only) then $\pi_{RD}$ selects the UE $\ell^\star\in\cI_R$ for which  $(\alpha_{\ell^\star} t)/p_{\ell^\star}-\hat{R}_{\ell^\star}(t)$ is maximum. Note that steps $1-2$ of $\pi_{RD}$ are identical to policy $\pi_H$, except that the packets for UEs $j\in\cI_L$ are not considered. 

In the final step (step $3$), $\pi_{RD}$ picks a UE $u$ from the set $\cI_L\cup\{\ell^\star\}$ following a randomized strategy, and transmits its latest packet. In particular, in any slot $t$, let $\cI_L^S(t)\subseteq$ denote the subset of UEs in $\cI_L$ for which there is a packet to transmit. 
$\pi_{RD}$ assigns probability $\theta_j$ \eqref{eq:theta-j} to each UE $j\in\cI_L^S(t)$, and probability $\theta_{\ell^\star}=1-\sum_{j\in\cI_L^S(t)}\theta_j$ to UE $\ell^\star$. Then, among all the UEs in set $\cI_L^S\cup\{\ell^\star\}$, $\pi_{RD}$ picks a UE $u$ with probability $\theta_u$, and transmits its latest packet.

The key property of policy $\pi_{RD}$ is as follows.
\begin{theorem}
	Under Assumption \ref{assume:sum-theta-less-than-1}, $\pi_{RD}$ satisfies the average latency constraint \eqref{eq:alt-constraint-latency} and throughput constraint \eqref{eq:alt-constraint-throughput}, for all UEs $j\in\cI_L$ and UEs $k\in\cI_R$, respectively.
\end{theorem}
\begin{proof}
	In any slot $t$, if BS has a packet for UE $j\in\cI_L$, then under $\pi_{RD}$, the BS gets to transmit the packet to UE $j$ with probability $\theta_j$. Also, the transmission succeeds with probability $p_\ell$. Therefore, for UE $j$, in each slot $t$ when there is a packet to transmit, the effective service rate is $\hat{p}_j=\theta_j p_j=q_j+(1-q_j)/\beta_j$. Therefore, using the result for Geo/Geo/1 queue (discussed above), the average latency for UE $j$ is $\beta_j$, thus satisfying the average latency constraint \eqref{eq:alt-constraint-latency}. 
	
	Further, similar to policy $\pi_H$ (and $\pi_{VW}$), only as many packets of UEs $\ell\in\cI_A\cup\cI_L$ are considered for transmission, such that sufficient slots are available to satisfy the throughput constraint of UEs $k\in\cI_R$. Also, in slots in which the UEs in $\cI_R$ are served, the strategy used to choose which particular UE gets served is identical to that under policy $\pi_H$ (and $\pi_{VW}$). Therefore, similar to policy $\pi_H$ (and $\pi_{VW}$), $\pi_{RD}$ satisfies the throughput constraint \eqref{eq:alt-constraint-throughput}, $\forall k\in\cI_R$. 
\end{proof}

\begin{remark} \label{remark:NC-latency-1UE}
	When there is a signal UE in $\cI_L$, assumption $1$ is in fact a necessary condition for feasibility of the average latency constraint \eqref{eq:alt-constraint-latency}. This is obvious because for such a case, $\sum_{j\in\cI_L}\theta_j=\theta_{L1}>1$, which implies that the effective service rate needed to satisfy the latency constraint needs to be $\hat{p}_{L1}=\theta_{L1} p_{L1}>1$, which is not possible.
\end{remark}

\begin{algorithm}
	\caption{Hierarchical Index Policy with Randomization.}
	\label{algo:alt-policy}
	\begin{algorithmic} [1]
		\STATE Initialize variables $W_i=0$, $a_i=0$, $b_i=0$ ($\forall i\in\cI_A$); 
		\STATE Compute $\avT_i^\star$ \eqref{eq:compute-Tstar} ($\forall i\in\cI_A$); 
		\FOR{each slot $t\in\{1,2,3,\cdots\}$}
		\STATE \textbf{Update\_Index\_UEs\_$\mathbf{\cI_A}$} (Codeblock \ref{algo:random-policy-index-UE-AoI});
		\STATE \textbf{Pick\_Best\_UE\_$\mathbf{\cI_A\cup\cI_R}$} (Codeblock \ref{algo:random-policy-pick-bestUE-AoI-Throughput});		
		\STATE \textbf{Transmit\_Random\_UE\_$\mathbf{\cI}$} (Codeblock \ref{algo:random-policy-randomized-transmission});	    
		\ENDFOR
	\end{algorithmic}
\end{algorithm}

\begin{algorithm}
	\floatname{algorithm}{Codeblock}
	\caption{\textbf{Update\_Index\_UEs\_$\mathbf{\cI_A}$}} 
\label{algo:random-policy-index-UE-AoI}
\begin{algorithmic} [1]
	\IF{a packet arrives for any UE $i\in\cI_A$} 
	\IF{$t-a_i\ge \lceil \avT_i^\star-1/q_i\rceil$}
	\STATE $a_i=t$, $b_i=b_i+1$;
	\ENDIF
	\IF{$b_i>|\cR_i(t-1)|$}
	\STATE $W_i=\rho_i p_i (t-\lambda_i(t))$;
	\ENDIF
	\ENDIF
\end{algorithmic}
\end{algorithm}

\begin{algorithm}
	\floatname{algorithm}{Codeblock}
	\caption{\textbf{Pick\_Best\_UE\_$\mathbf{\cI_A\cup\cI_R}$}} 
\label{algo:random-policy-pick-bestUE-AoI-Throughput}
\begin{algorithmic} [1]
	\IF{$\max_{i\in\cI_A} W_i>0$} 
	\STATE $\ell^\star=\arg\max_{i\in\cI_A} W_i$; 
	\ELSE \STATE $\ell^\star=\arg\max_{k\in\cI_R}\left\{\frac{\alpha_k\cdot t}{p_k}-\hat{R}_k(t)\right\}$;
	\ENDIF
\end{algorithmic}
\end{algorithm}

\begin{algorithm}
	\floatname{algorithm}{Codeblock}
	\caption{\textbf{Transmit\_Random\_UE\_$\mathbf{\cI}$}} 
\label{algo:random-policy-randomized-transmission}
\begin{algorithmic} [1]
	\STATE $\cI_L^S=\{j\in\cI_L| \text{BS has a packet to transmit to UE $j$}\}$;
	\STATE compute $\theta_j$ \eqref{eq:theta-j}, $\forall j\in\cI_L^S$;
	\STATE $\theta_{\ell^\star}=1-\sum_{j\in\cI_L^S}\theta_j$;
	\STATE among all UEs in $\cI_L^S\cup\{\ell^\star\}$, choose a UE $u$ with probability $\theta_u$, and transmit its latest packet;
	\IF{$u\in\cI_A$ \AND transmission successful}
	\STATE $W_u=0$; 
	\ENDIF
\end{algorithmic}
\end{algorithm}

\section{Numerical Results} \label{sec:numerical-results}

In this section, we analyze the performance of the proposed policies using numerical simulation, and organize the discussion into two parts. In first part, we analyze policy $\pi_H$ (Algorithm \ref{algo:policy}) as a solution of problem \eqref{eq:prob-form-1}, and subsequently in the second part, we analyze policies $\pi_{VW}$ and $\pi_{RD}$ as a solution of problem \eqref{eq:prob-form-2}. 
For both the parts, we consider a system with three UEs: i) UE $1\in\cI_A$, ii) UE $2\in\cI_L$ and iii) UE $3\in\cI_R$, and fix the following simulation parameters: arrival rate $[q_1,q_2,q_3]=[0.9,0.2,1]$, transmission success probability $[p_1,p_2,p_3=0.7,0.8,0.9]$, weights $[\rho_1,\rho_2]=[1,1]$, and constants $[f,\eta]=[10000,0.1]$.

\subsection{Results for Policy $\pi_H$ (Algorithm \ref{algo:policy})} \label{sec:numerical-results-piH}
First, we simulate policy $\pi_H$ for different values of the throughput constraint $\alpha=\alpha_3$ (for UE $3\in\cI_R$). Figures \ref{fig:Throughput_vs_alpha} shows the plot of the throughput of the three UEs as a function of $\alpha$. From Figure \ref{fig:Throughput_vs_alpha}, we make the following three key observations. 

\begin{figure}[htbp]
	\centerline{\includegraphics[scale=0.5]{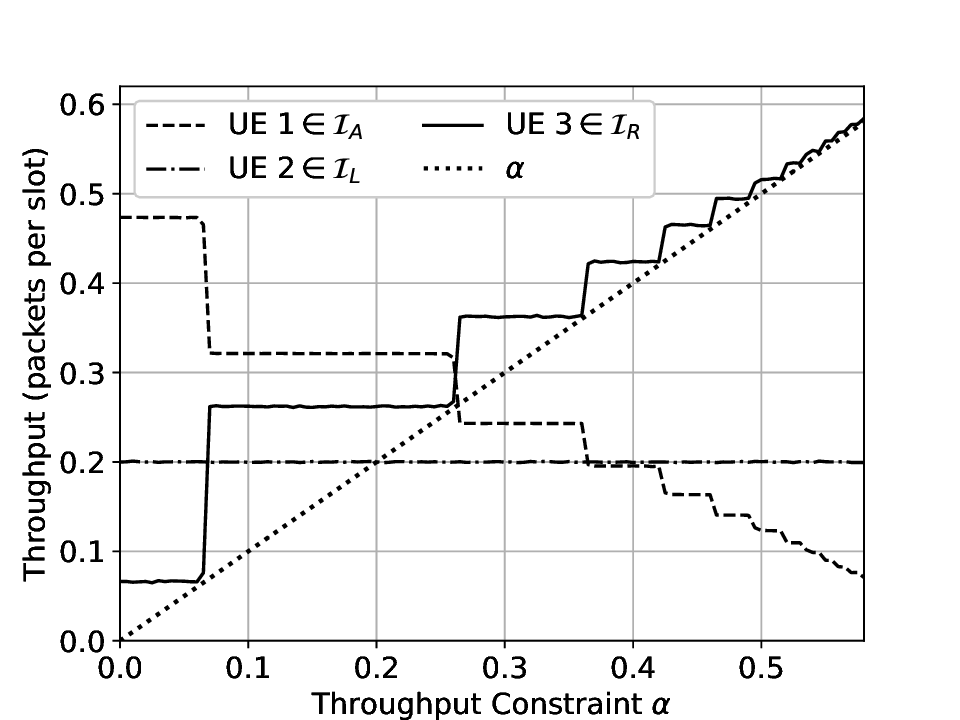}}
	\caption{Dependence of UE throughputs on constraint $\alpha$, under policy $\pi_H$.}
	\label{fig:Throughput_vs_alpha}
\end{figure}

	1) The throughput of UE $2\in\cI_L$ does not depend on $\alpha$, and is equal to the packet arrival rate $q_2=0.2$. This is expected because all packets for UE $2\in\cI_L$ must be transmitted (Lemma \ref{lemma:throughput-UEs}), and $\pi_H$ does exactly that. 
	
	2) Throughput $\avR_3^H$ for UE $3\in\cI_R$ increases with increase in $\alpha$. 
	This is because as shown in Theorem \ref{thm:throughput-constraint-satisfied-piH}, $\pi_H$ satisfies the throughput constraint \eqref{eq:constraint}, which mandates that $\avR_3^H\ge \alpha$. Also, the throughput $\avR_1^H$ for UE $1\in\cI_A$ decreases with increase in $\alpha$. This is expected because $\pi_H$ is a work-conserving policy that in each slot, serves some UE. Hence, throughput of some UE (UE $1$) must decrease to support higher throughput of other UE (UE $3$ in this case). The relation between the throughput of different UEs is captured in \eqref{eq:throughput-relation} (Appendix \ref{app:proof-lemma-throughput-UEs}).
	
	3) Under $\pi_H$, the throughputs $\avR_1^H$ and $\avR_3^H$ varies with $\alpha$ in steps. Also, the size of the steps decreases with increase in $\alpha$. The reason for the step-wise decrease in $\avR_1^H$ (increase in $\avR_3^H$) is the the choice of threshold $\lceil \avT_1^\star-1/q_1\rceil$ (used by $\pi_H$ for incrementing counter $b_1$) that varies step-wise with $\avT_1^\star$. Since $\avT_1^\star$ satisfies \eqref{eq:avT-constraint} (by definition), as $\alpha$ increases, $\avT_1^\star$ increases (at an increasing rate). Therefore, the average inter-arrival time of packets $\avT_1^H$ transmitted to UE $1\in\cI_A$, which is proportional to $\lceil \avT_1^\star-1/q_1\rceil$, increases step-wise. Therefore, $\avR_1^H=1/\avT_1^\star$ (Lemma \ref{lemma:relation-throughput-cycle-length}) decreases step-wise. Further, since 	$\avT_1^\star$ increases at an increasing rate with $\alpha$, therefore, $\lceil \avT_1^\star-1/q_1\rceil$ increases ($\avR_1^H$ decreases) in steps with decreasing step size.

Next, for the same simulation, Figure \ref{fig:Cost_vs_alpha} shows the plot of average AoI $\avA_1^H$ of UE $1\in\cI_A$, average latency $\avL_2^H$ of UE $2\in\cI_L$, and the total cost $\avA_1^H+\avL_2^H$. 
Recall that the throughput $\avR_1^H$ (for UE $1\in\cI_A$) decreases with increase in $\alpha$. Therefore, the average inter-transmission time of packets to UE $1$ increases, thereby increasing the average AoI $\avA_1^H$.  
Further, as $\alpha$ increases, fewer packets of UE $1\in\cI_A$ are added to set $\cS$ of high priority packets. Therefore, packets of UE $2\in\cI_L$ in set $\cS$ have to wait lesser, and hence, as $\alpha$ increases, the average latency $\avL_2^H$ of UE $2$ decreases. 

\begin{figure}[htbp]
	\centerline{\includegraphics[scale=0.5]{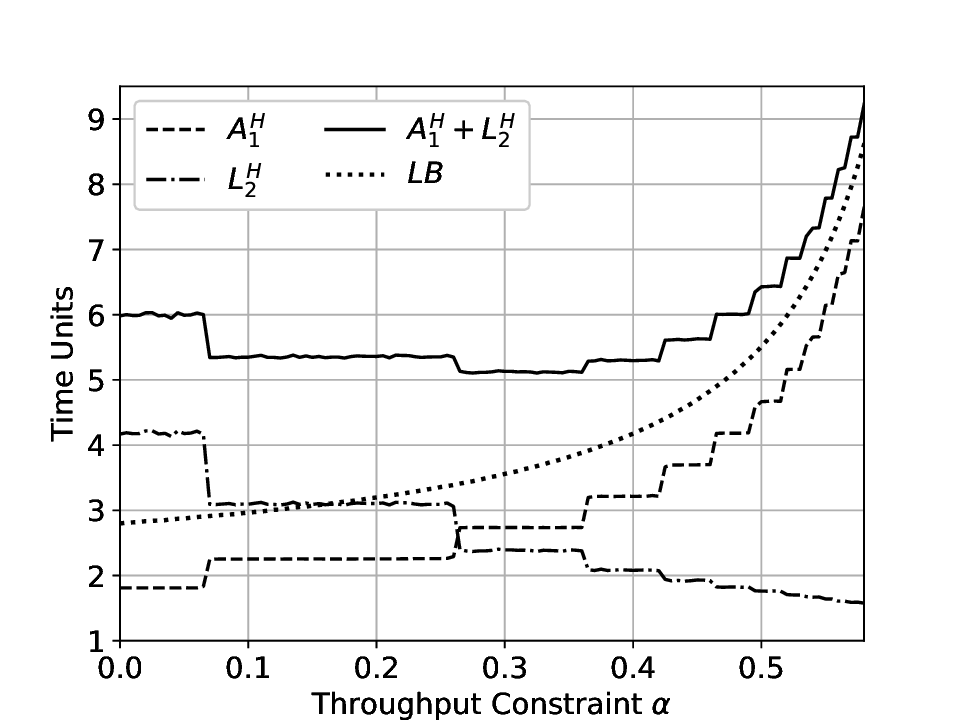}}
	\caption{Average AoI and average latency as a function of throughput constraint $\alpha$.}
	\label{fig:Cost_vs_alpha}
\end{figure}

Further, to compare the total cost \eqref{eq:objective-equiv} for $\pi_H$ with optimal cost, Figure \ref{fig:Cost_vs_alpha} shows a lower bound ($LB$) on the total cost \eqref{eq:objective-equiv} across all policies $\pi\in\Pi_F$, that is given as follows. 
\begin{align} \label{eq:LB}
	LB=\min_{\pi\in\Pi_F} F_1^\pi+\min_{\pi\in\Pi_F} F_2^\pi \ge LB_{F1} + LB_{F2}, 
\end{align}
where $LB_{F1}= \min_{\avT_1} \frac{1}{2}\left(\avT_1 +\frac{(1-q_1)/(2q_1)}{\avT_1}+1\right) \text{ (s.t. \eqref{eq:avT-constraint})}$ is a lower bound on $F_1^\pi$, $\forall \pi\in\Pi_F$ (follows from \eqref{eq:F1-lb}), and $LB_{F2}=\min_{\pi\in\Pi_F}\underset{t\to\infty}{\lim} \sum_{\stackrel{\ell\in\cI_L}{i\in\cG_\ell(t)}} (\rho_\ell/q_\ell)L_{\ell i}^\pi/t$ is a lower bound on $F_2^\pi$, $\forall \pi\in\Pi_F$ (since $\avL_{\ell i}^\pi\ge 1$, by definition). 
\begin{remark}
	We obtained $LB_{F1}$ by solving the convex optimization problem in its definition, while $LB_{F2}$ is obtained using simulation, by applying the $c\mu$-rule \cite{buyukkoc1985cmu}, which is known to be optimal for this case. 
\end{remark}
Note that because packets for UE $1\in\cI_A$ are not considered while deriving $LB_{F2}$, in general, $LB$ is a loose lower bound on the optimal cost. Therefore, $LB$ should be much smaller than the total cost \eqref{eq:objective-equiv} for $\pi_H$. However, as $\alpha$ increases, $\pi_H$ transmits fewer packets to UE $1$, and hence,  $F_2^H$ approaches $LB_{F2}$, and the total cost \eqref{eq:objective-equiv} for $\pi_H$ approaches $LB$ (as shown in Figure \ref{fig:Cost_vs_alpha}). 
This suggests that with increase in $\alpha$, $\pi_H$ approaches optimality.
\begin{remark}
	In practice, $\pi_H$ may be close to optimal even for small $\alpha$'s. However, we cannot prove/disprove this fact using $LB$, as for small $\alpha$, $LB$ is a very poor estimate of optimal cost (for reason discussed above). 
\end{remark}

\subsection{Results for Policies $\pi_{VW}$ and $\pi_{RD}$ (Algorithms \ref{algo:find-weights} and \ref{algo:alt-policy})} 
Further, we consider the optimization problem \eqref{eq:prob-form-2}, and simulate policies $\pi_{VW}$ and $\pi_{RD}$ with throughput constraint $\alpha=\alpha_3=0.2$, for different values of the latency constraint $\beta=\beta_2$ (for UE $2\in\cI_L$). Figures \ref{fig:Weight_vs_Frame}--\ref{fig:Alt_Throughput} reveal some key results from the simulation, discussed next.

Recall that policy $\pi_{VW}$ differs from $\pi_H$ only in one step where it updates the virtual weights. Therefore, to analyze the effectiveness of this step, Figure \ref{fig:Weight_vs_Frame} plots the virtual weight $\rho_2$ against the number of update iterations, for different values of $\beta$. Note that for larger $\beta$, $\rho_2$ converges to a smaller value. When $\beta$ is too small (e.g. $\beta=1$, for which the latency constraint \eqref{eq:alt-constraint-latency} is infeasible), $\rho_2$ grows to infinity, while for large $\beta$, $\rho_2$ converges towards $0$. This behavior is expected as for smaller $\beta$, packets of UE $2\in\cI_L$ needs to be given higher priority over packets for UE $1\in\cI_A$. On the other hand, when $\beta$ is large, UE $2\in\cI_L$ can tolerate larger latency, and hence, packets of UE $1\in\cI_A$ is given higher priority so that its average AoI can be minimized. 
\begin{figure}[htbp]
	\centerline{\includegraphics[scale=0.5]{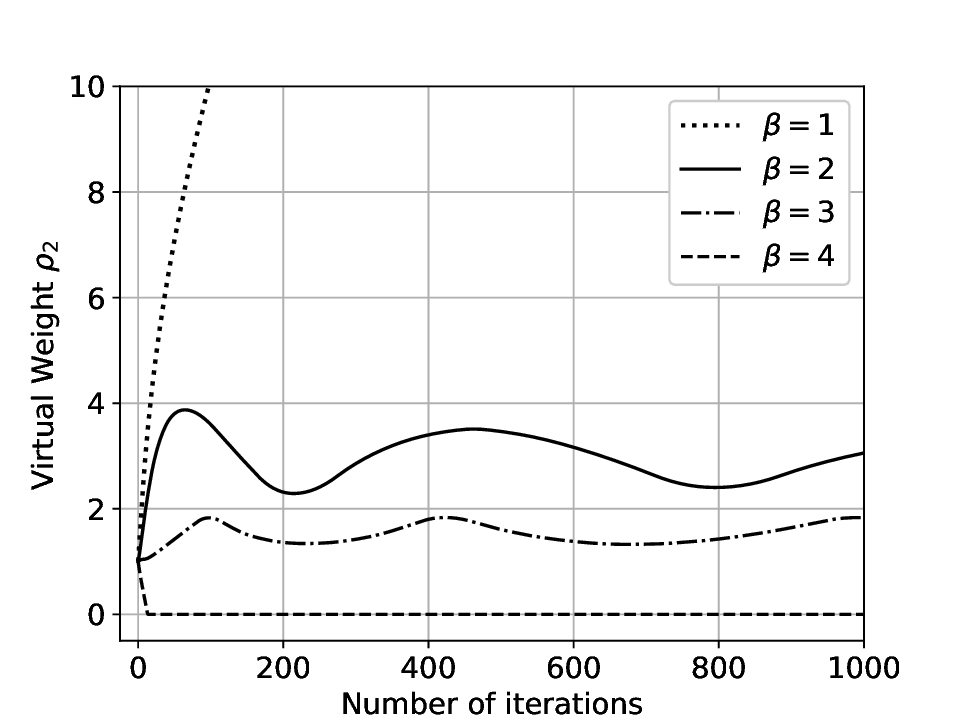}}
	\caption{Virtual weight $\rho_2$ as a function of number of iterations.}
	\label{fig:Weight_vs_Frame}
\end{figure}

Figure \ref{fig:Alt_AoI_Latency} shows the plot of average AoI $A_1^\pi$ (of UE $1\in\cI_A$) and average latency $L_2^\pi$ (of UE $2\in\cI_L$) under policies $\pi_{VW}$ and $\pi_{RD}$. From Remark \ref{remark:NC-latency-1UE} we know that for the considered system, when $\beta< 1.33$ the problem is infeasible as the average latency for UE $2\in\cI_L$ cannot be less than $1.33$. Therefore, when $\beta< 1.33$, the average latency $\avL_2^\pi=1.33$ for both $\pi_{VW}$ and $\pi_{RD}$. When $\beta\ge 1.33$, as $\beta$ increases, the priority (virtual weight $\rho_2$ under $\pi_{VW}$, and probability $\theta_2$ under $\pi_{RD}$) of the packets of UE $2\in\cI_L$ (compared to UE $1\in\cI_A$) is lowered such that they are just enough to satisfy the latency constraint \eqref{eq:alt-constraint-latency}. Hence, as $\beta$ increases, there is relative increase in the priority of packets of UE $1\in\cI_A$, thus lowering its average AoI. 

Apart from the commonalities, there are two major differences in the performance of policies $\pi_{VW}$ and $\pi_{RD}$ that can be observed from Figure \ref{fig:Alt_AoI_Latency}

	1) Beyond a certain value of $\beta$, under $\pi_{VW}$, the average latency $\avL_2$ (of UE $2$) stops increasing. This is because, under $\pi_{VW}$, the latency of UE $2\in\cI_L$ increases only if there is some packet of UE $1\in\cI_A$ of higher priority. But the number of packets of UE $1\in\cI_A$ are limited. Therefore, for large $\beta$, $\avL_2$ saturates (and virtual weight $\rho_2=0$; see Figure \ref{fig:Weight_vs_Frame}).
	
	In contrast, under $\pi_{RD}$, packet of UE $2\in\cI_L$ is transmitted with a fixed probability $\theta_2$ \eqref{eq:theta-j}. Therefore, even when there is no packet for UE $1\in\cI_A$, $\pi_{RD}$ may not transmit the packet of UE $2\in\cI_L$ (with probability $1-\theta_2$), and instead, transmit a packet of UE $3\in\cI_R$. Therefore, under $\pi_{RD}$, the average latency $\avL_2^{RD}$ increases linearly with $\beta$, without saturating.
	
	2) As $\beta$ increases (and policies $\pi_{VW}$ and $\pi_{RD}$ allow the average latency $L_2^\pi$ to increase), the average AoI decreases. However, the decrease is more evident in $\pi_{VW}$ because as $\beta$ increases, virtual weight $\rho_2$ (for UE $2\in\cI_L$) decreases, and hence, the packets of UE $1\in\cI_A$ are always preferred  over the packets of UE $2\in\cI_L$ for transmission. In contrast, $\pi_{RD}$ is a randomized algorithm, and even when $\beta$ is large, $\pi_{RD}$ may choose to transmit packets of UE $2\in\cI_L$ before packets of UE $1\in\cI_A$ with probability $\theta_2$ \eqref{eq:theta-j} (although $\theta_2$ decreases with increase in $\beta$, it is always positive). 

\begin{figure} 
	\centerline{\includegraphics[scale=0.5]{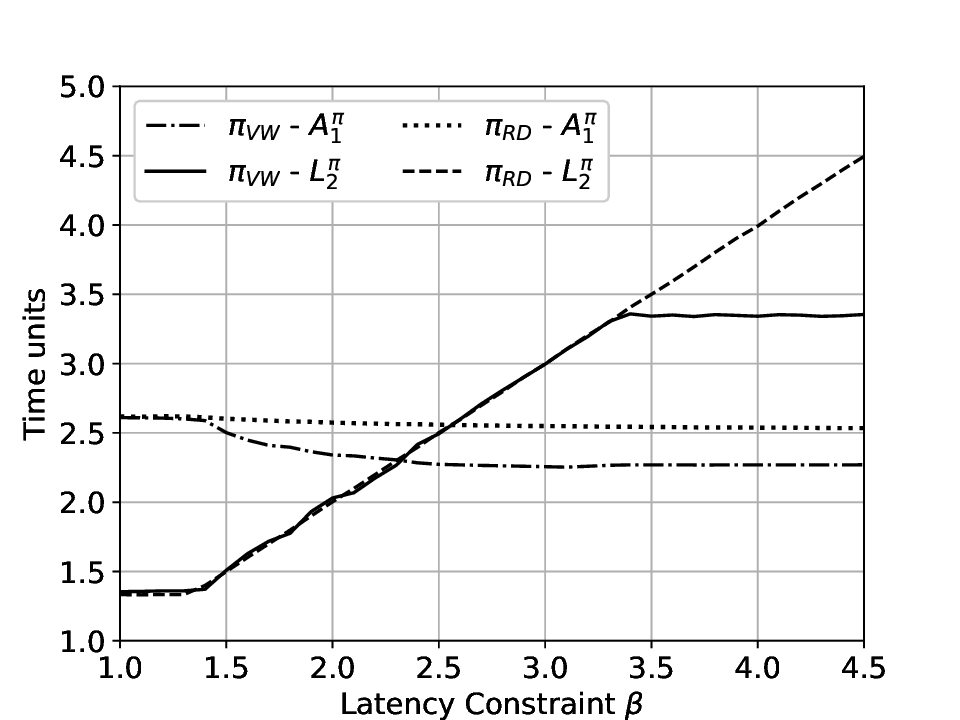}}
	\caption{Average AoI and average latency as a function of latency constraint $\beta$.}
	\label{fig:Alt_AoI_Latency}
\end{figure}
\begin{remark}
	In Figure \ref{fig:Alt_AoI_Latency}, note that under policy $\pi_{VW}$, the average latency $L_2^\pi$ increases linearly with increase in $\beta\in[2.5,3.8]$. In contrast, $\beta\in[2.5,3.8]$ has no significant effect on average AoI $A_1^\pi$. Therefore, for policy $\pi_{VW}$, the suitable operating point $\beta^{VW}\le 2.5$. 
\end{remark}

Finally, Figure \ref{fig:Alt_Throughput} shows that the throughput of the UEs do not depend on $\beta$, and are equal for both policies $\pi_{VW}$ and $\pi_{RD}$. This is expected, because $\beta$ only affects the order in which the packets of different UEs are transmitted, and not the number of transmitted packets.
\begin{figure}[htbp]
	\centerline{\includegraphics[scale=0.5]{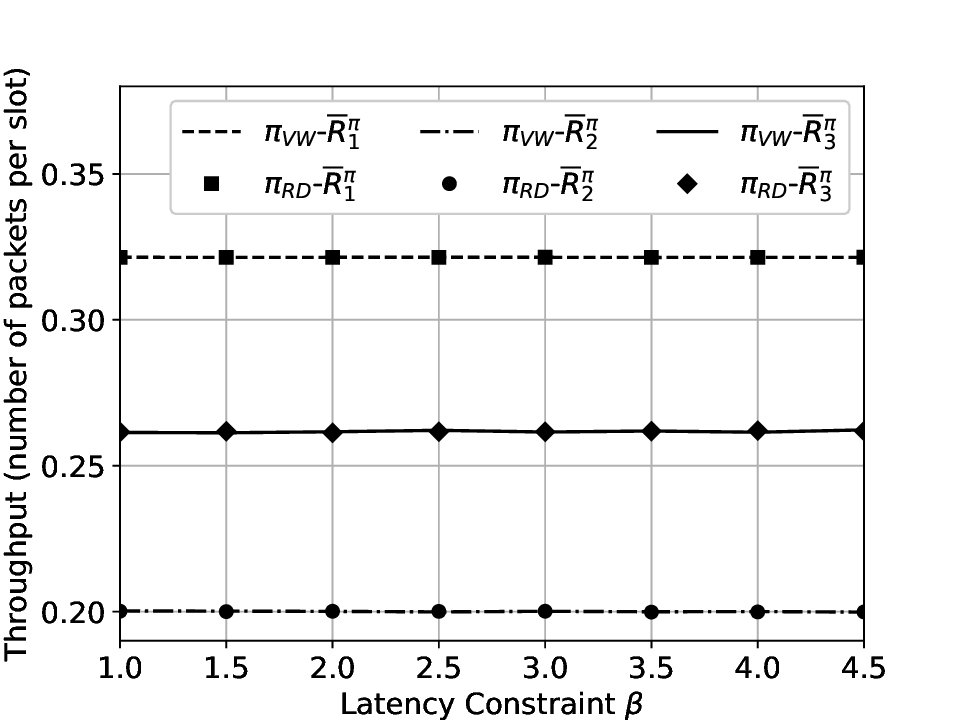}}
	\caption{Throughput of UEs as a function of latency constraint $\beta$.}
	\label{fig:Alt_Throughput}
\end{figure}

\section{Conclusion} \label{sec:conclusion}
We considered a base station (BS) that serves a set of heterogeneous users (UE) that may either require high throughput, low latency, or low age-of-information (AoI). 
In each time slot, the BS needs to choose which UE to serve such that the long-term QoS requirements of all the UEs are fulfilled. Hence, the BS needs to handle a three-way tradeoff between throughput, latency and AoI of the UEs. In this paper, we captured these tradeoffs via an optimization problem, and proposed candidate solutions (algorithms) for this problem. Further, we analyzed the algorithms both analytically and through simulations, and found that a hierarchical index based approach provides a general framework to handle the service requirements of UEs in terms of throughput, latency and AoI.

\appendices
\section{Proof of Lemma \ref{lemma:throughput-UEs}} \label{app:proof-lemma-throughput-UEs}

	1) For policies $\pi\in\Pi_F$, the average AoI $\avA_i^\pi$ (for UE $i\in\cI_A$) is finite, which is possible only if UE $i$ periodically receives new packets. Hence, the throughput $\avR_i^\pi$ must be strictly positive, $\forall i\in\cI_A$.
	
	2) From Remark \ref{remark:latency-tx-every-packet}, we know that the BS must transmit all the packets that arrive for UE $j\in\cI_L$.  
	Also, for UE $j\in\cI_L$, packets arrive at the BS at rate $q_j$. Therefore, for UE $j\in\cI_L$, the throughput \eqref{eq:throughput} must be $\avR_j^\pi=q_j$.
	
	3) Since $\pi\in\Pi_F$ satisfy the constraint \eqref{eq:constraint}, the throughput for UEs $k\in\cI_R$ must be at least $\alpha_k$. 
	
	4) For policies $\pi\in\Pi_F$ and UE $\ell\in\cI$, let $\hat{R}_\ell^\pi(t)$ and $R_\ell^\pi(t)=|\cR_\ell^\pi(t)|$ respectively denote the number of slots in which the BS transmits to UE $\ell$, and the number of slots in which the transmission is successful (i.e., the number of packets received at UE $\ell$).  
	Recall that whenever a packet is transmitted to UE $\ell\in\cI$, the transmission succeeds with a fixed probability $p_\ell$ (independent of everything else), and fails otherwise. Therefore, when BS transmits to UE $\ell$ in $\hat{R}_\ell^\pi(t)$ number of slots, the expected number of successful transmissions is $\bbE[R_\ell^\pi(t)]=p_\ell\hat{R}_\ell^\pi(t)$. Also, since $q_\ell,\alpha_k>0$ $\forall \ell\in\cI_L,\cI_R$, and the average AoI $\avA_\ell^\pi$'s are finite for policies $\pi\in\Pi_F$, as time $t\to\infty$, BS must transmit to each UE $\ell\in\cI=\cI_A\cup\cI_L\cup\cI_R$ in infinitely many slots. Therefore, using the strong law of large numbers, we get that for each UE $\ell\in\cI$, as $t\to\infty$, $R_\ell^\pi(t)=p_\ell\hat{R}_\ell^\pi(t)$ with probability $1$. Hence, $\forall \ell\in\cI$, the throughput $\avR_\ell^\pi=\lim_{t\to\infty}R_\ell^\pi(t)/t=\lim_{t\to\infty}p_\ell\hat{R}_\ell^\pi(t)/t$, which implies
	\begin{align} \label{eq:throughput-equiv}
		\frac{\avR_\ell^\pi}{p_\ell}=\lim_{t\to\infty}\frac{\hat{R}_\ell^\pi(t)}{t}.
	\end{align} 
	
	Further, recall that in any slot, BS may transmit to at most one UE. Therefore, under any policy $\pi\in\Pi_F$, $\sum_{i\in\cI_A}\hat{R}_i^\pi(t)+\sum_{j\in\cI_L}\hat{R}_j^\pi+\sum_{k\in\cI_R}\hat{R}_k^\pi\le t$. Dividing both sides by $t$, and using \eqref{eq:throughput-equiv}, we get
	\begin{align} \label{eq:throughput-relation}	\sum_{i\in\cI_A}\frac{\avR_i^\pi}{p_i}+\sum_{j\in\cI_L}\frac{\avR_j^\pi}{p_j}+\sum_{k\in\cI_R}\frac{\avR_k^\pi}{p_k}\le 1,
	\end{align}
	which together with the Results $2$ and $3$ of Lemma \ref{lemma:throughput-UEs}, implies 	$\sum_{i\in\cI_A}\avR_i^\pi/p_i\le 1-\sum_{j\in\cI_L}q_j/p_j-\sum_{k\in\cI_R}\alpha_k/p_k$.

\section{Proof of Lemma \ref{lemma:relation-throughput-cycle-length}} \label{app:proof-lemma-relation-throughput-cycle-length} 

As argued in the proof of Lemma \ref{lemma:throughput-UEs} (Result $4$; Appendix \ref{app:proof-lemma-throughput-UEs}), for any policy $\pi\in\Pi_F$, each UE $\ell\in\cI$ transmits infinitely many packets. Therefore, time $t$ can be expressed as the sum of inter-generation time of successive packets transmitted by any given UE $\ell\in\cI$, i.e., $t=\sum_{i\in\cR_\ell^\pi(t)}T_{\ell i}^\pi$. Therefore, substituting $t=\sum_{i\in\cR_\ell^\pi(t)}T_{\ell i}^\pi$ into \eqref{eq:throughput} and using \eqref{eq:avT}, we get the result.

\section{Proof of Lemma \ref{lemma:avAoI-general-expression}} \label{app:proof-lemma-avAoI-general-expression}

    Let $\cR_\ell^\pi=\underset{t\to\infty}{\lim}\cR_\ell^\pi(t)$.
	By definition, $\delta_{\ell i}^\pi=T_{\ell i}^\pi-\avT_\ell^\pi$, i.e., $T_{\ell i}^\pi=\avT_\ell^\pi+\delta_{\ell i}^\pi$. Therefore, 
	$(T_{\ell i}^\pi)^2=(\avT_\ell^\pi+\delta_{\ell i}^\pi)^2=(\avT_\ell^\pi)^2+(\delta_{\ell i}^\pi)^2+2\avT_\ell^\pi\delta_{\ell i}^\pi$. Taking the sum over $i\in\cR_\ell^\pi$, we get
	\begin{align} \label{eq:sum-T2}
		\sum_{i\in\cR_\ell^\pi}(T_{\ell i}^\pi)^2&=\sum_{i\in\cR_\ell^\pi}(\avT_\ell^\pi)^2+\sum_{i\in\cR_\ell^\pi}(\delta_{\ell i}^\pi)^2+\sum_{i\in\cR_\ell^\pi}2\avT_\ell^\pi\delta_{\ell i}^\pi, \nonumber \\
		&=|\cR_\ell^\pi|(\avT_\ell^\pi)^2+\sum_{i\in\cR_\ell^\pi}(\delta_{\ell i}^\pi)^2+2\avT_\ell^\pi\sum_{i\in\cR_\ell^\pi}\delta_{\ell i}^\pi, \nonumber \\
		&\stackrel{(a)}{=}|\cR_\ell^\pi|(\avT_\ell^\pi)^2+\sum_{i\in\cR_\ell^\pi}(\delta_{\ell i}^\pi)^2,
	\end{align}
where we get $(a)$ because $\avT_\ell^\pi$ being an average of $T_{\ell i}^\pi$'s ($\forall i\in\cR_\ell^\pi$), 	$\sum_{i\in\cR_\ell^\pi}\delta_{\ell i}^\pi=\sum_{i\in\cR_\ell^\pi}(T_{\ell i}^\pi-\avT_\ell^\pi)=0$.

Consider the following well-known result derived in  \cite{kadota2019minimizing}. 
\begin{lemma} \label{lemma:avAoI-general-expression-dummy}
	[Proposition $2$ in \cite{kadota2019minimizing}] Under any policy $\pi\in\Pi_F$, for each UE $\ell\in\cI_A$, the average AoI 
	\begin{align} \label{eq:avAoI-equiv-dummy}
		\avA_\ell^\pi=\lim_{t\to\infty}\frac{1}{t}\sum_{i\in\cR_\ell^\pi(t)}\left(\frac{(T_{\ell i}^\pi)^2}{2}+T_{\ell i}^\pi (L_{\ell i}^\pi-1)\right)+\frac{1}{2},
	\end{align}
	with probability $1$. 
\end{lemma}

Substituting \eqref{eq:sum-T2} into \eqref{eq:avAoI-equiv-dummy}, and because $\underset{t\to\infty}{\lim}\sum_{i\in\cR_\ell^\pi(t)}T_{\ell i}^\pi=t$ (shown in the Proof of Lemma \ref{lemma:relation-throughput-cycle-length}; Appendix \ref{app:proof-lemma-relation-throughput-cycle-length}), we get 
\begin{align} \label{eq:avAoI-dummy-1}
	\avA_\ell^\pi-\lim_{t\to\infty}&\sum_{i\in\cR_\ell^\pi(t)}\frac{T_{\ell i}^\pi (L_{\ell i}^\pi-1)}{t}-\frac{1}{2} \nonumber \\
	&=\frac{|\cR_\ell^\pi|(\avT_\ell^\pi)^2+\sum_{i\in\cR_\ell^\pi}(\delta_{\ell i}^\pi)^2}{2\cdot \sum_{i\in\cR_\ell^\pi(t)}T_{\ell i}^\pi}, \nonumber \\ 
	&=\frac{(\avT_\ell^\pi)^2+\sum_{i\in\cR_\ell^\pi}(\delta_{\ell i}^\pi)^2/|\cR_\ell^\pi|}{2\cdot \sum_{i\in\cR_\ell^\pi}T_{\ell i}^\pi/|\cR_\ell^\pi|}, \nonumber \\ 
	&\stackrel{(a)}{=}\frac{\avT_\ell^\pi}{2}+\frac{\avdelta}{2 \avT_\ell^\pi}, 
\end{align}
where we get $(a)$ using \eqref{eq:avT}, and defining $\avdelta=\sum_{i\in\cR_\ell^\pi}(\delta_{\ell i}^\pi)^2/|\cR_\ell^\pi|$. 
Finally, rearranging the terms on the two sides of \eqref{eq:avAoI-dummy-1}, 
we get \eqref{eq:avAoI-equiv}.

\section{Proof of Lemma \ref{lemma:avdelta-lb}} \label{app:proof-avdelta-lb}

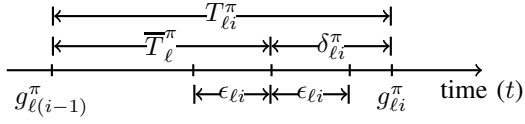
\begin{figure} 
	\begin{center}
		\begin{tikzpicture}[thick,scale=0.8, every node/.style={scale=1}]
		\draw[->] (-0.4,0) to (7.5,0) node[below]{time ($t$)};

		\draw (0.35,-0.1) node[below]{$g_{\ell (i-1)}^\pi$} to (0.35,0.1);
		\draw (2.7,-0.1) to (2.7,0.1);
		\draw (4,-0.1) to (4,0.1);
        \draw (5.3,-0.1) to (5.3,0.1); 
        \draw (6,-0.1) node[below]{$g_{\ell i}^\pi$} to (6,0.1);

		\draw[|<->|] (0.35,0.4) -- (4,0.4) node[rectangle,inner sep=-1pt,midway,fill=white]{$\avT_\ell^\pi$}; 
		\draw[|<->|] (0.35,0.9) -- (6,0.9) node[rectangle,inner sep=-1pt,midway,fill=white]{$T_{\ell i}^\pi$};
		\draw[<->|] (4,0.4) -- (6,0.4) node[rectangle,inner sep=-1pt,midway,fill=white]{$\delta_{\ell i}^\pi$};
		\draw[|<->|] (2.7,-0.4) -- (4,-0.4) node[rectangle,inner sep=-1pt,midway,fill=white]{$\epsilon_{\ell i}$}; 
		\draw[<->|] (4,-0.4) -- (5.3,-0.4) node[rectangle,inner sep=-1pt,midway,fill=white]{$\epsilon_{\ell i}$};

		\end{tikzpicture}
		\caption{If no packet is arrives in slots $(g_{\ell (i-1)}^\pi+(\avT_{\ell}^\pi-\epsilon_{\ell i})$ to $g_{\ell (i-1)}^\pi+(\avT_{\ell}^\pi+\epsilon_{\ell i})$, then $(\delta_{\ell i}^\pi)^2=(T_{\ell i}^\pi-\avT_{\ell}^\pi)^2\ge\epsilon_{\ell i}^2$.\vspace{-2ex}}  
		\label{fig:cycle-var} 
	\end{center}
\end{figure}
Note that under any policy that minimizes \eqref{eq:objective-equiv}, whenever BS transmits a packet to UE $\ell\in\cI_A$, it is the latest among all packets that arrived for UE $\ell$. This is because the transmission success probability is equal for all the packets for UE $\ell$, and transmitting latest packet results in greater reduction in AoI. Also, upon transmitting a new packet, the BS discards all the previously arrived packets for UE $\ell$ (i.e. never transmits them), as AoI only depends on the generation time of the latest packet that arrived at the BS, and has been received at the UE. Therefore, if BS successfully transmits a packet that arrived in slot $g_{\ell (i-1)}^\pi$, to UE $\ell$ in slot $r_{\ell i}^\pi\ge g_{\ell i}^\pi$, then the next packet that the BS transmits to UE $\ell$, must arrive in some slot $g_{\ell i}^\pi>r_{\ell (i-1)}^\pi$.
Since policy $\pi\in\Pi_F$ only has causal information, this implies that in slot $r_{\ell (i-1)}^\pi$, $\pi$ does not know the arrival time of the next packet, in particular, $g_{\ell i}^\pi$. Therefore, if no packet arrives in slots $g_{\ell (i-1)}^\pi+(\avT_\ell^\pi-\epsilon_{\ell i})$ to $g_{\ell (i-1)}^\pi+(\avT_\ell^\pi+\epsilon_{\ell i})$ (for non-negative integer $\epsilon_{\ell i}$), then $T_{\ell i}^\pi=g_{\ell i}^\pi-g_{\ell (i-1)}^\pi$ will either be smaller than $\avT_\ell^\pi-\epsilon_{\ell i}$, or greater than $\avT_\ell^\pi+\epsilon_{\ell i}$, independent of everything else. 
In other words, $(\delta_{\ell i}^\pi)^2>\epsilon_{\ell i}^2$, where $\epsilon_{\ell i}$'s are independent of each other, $\forall i\ge 1$. Hence, 
\begin{align} \label{eq:avdelta-lb-1}
	\avdelta=\frac{\sum_{i\in\cR_\ell^\pi}(\delta_{\ell i}^\pi)^2}{|\cR_\ell^\pi|}> \frac{\sum_{i\in\cR_\ell^\pi}\epsilon_{\ell i}^2}{|\cR_\ell^\pi|}.
\end{align}

Recall that in each slot, a packet arrives for UE $\ell$ (at the BS) with probability $q_\ell$ (does not arrive with probability $1-q_\ell$), independent of everything else. Therefore, the probability that no packet arrives for UE $\ell$ in slots $g_{\ell (i-1)}^\pi+(\avT_\ell^\pi-\epsilon_{\ell i})$ to $g_{\ell (i-1)}^\pi+(\avT_\ell^\pi+\epsilon_{\ell i})$ is $(1-q_\ell)^{2\epsilon_{\ell i}+1}$, $\forall i\ge 1$. Combining this with the fact that $\epsilon_{\ell i}$'s are independent $\forall i\ge 1$, we get that $\epsilon_{\ell i}$'s are independent and identically distributed random variables. Also, the number of transmitted packets $|\cR_\ell^\pi|\to\infty$ (as argued in the proof of Result $4$ of Lemma \ref{lemma:throughput-UEs}; Appendix \ref{app:proof-lemma-throughput-UEs}). Thus, applying the strong law of large numbers to \eqref{eq:avdelta-lb-1}, we get  
$\begin{aligned} 
	\avdelta\ge \bbE[\epsilon_{\ell i}^2]\stackrel{(a)}{>}(1-q_\ell)/2q_\ell^2,
\end{aligned}$
where $\bbE[\cdot]$ is the expectation operator, and $(a)$ follows from the following lemma.

\begin{lemma}
	$\bbE[\epsilon_{\ell i}^2]> (1-q_\ell)/2q_\ell^2$, $\forall \ell\in\cI_A$, and $i\ge 1$.
\end{lemma}
\begin{proof}
	$\bbE[\epsilon_{\ell i}^2]=\sum_{x=0}^\infty \bbP(\epsilon_{\ell i}^2>x)$.
	Since $\epsilon_{\ell i}$ is an integer, the probability that $\epsilon_{\ell i}^2>x$, i.e. no packet arrives for UE $\ell$ in slots $g_{\ell (i-1)}^\pi+(\avT_\ell^\pi-\sqrt{x})$ to $g_{\ell (i-1)}^\pi+(\avT_\ell^\pi+\sqrt{x})$ is $\bbP(\epsilon_{\ell i}^2>x)\ge (1-q_\ell)^{2\lfloor \sqrt{x} \rfloor+1}$. Thus, 
	\begin{align} \label{eq:Exp-eps2}
		\bbE[\epsilon_{\ell i}^2]&\ge\sum_{x=0}^\infty (1-q_\ell)^{2\lfloor \sqrt{x} \rfloor+1}, \nonumber \\
		&=\sum_{x=0}^\infty (1-q_\ell)^{2x+1}((x+1)^2-x^2), \nonumber \\
		&=\sum_{x=0}^\infty (1-q_\ell)^{2x+1}(2x+1). 
	\end{align}
Computing the sum of the arithmetico-geometric progression in \eqref{eq:Exp-eps2}, we get 
$\begin{aligned}
\bbE[\epsilon_{\ell i}^2]\ge \frac{1-q_\ell}{q_\ell^2}\cdot\frac{2+q_\ell^2-2q_\ell}{4+q_\ell^2-4q}>\frac{1-q_\ell}{2q_\ell^2}.  
\end{aligned}$
\end{proof}


\section{Proof of Theorem \ref{thm:F1H-ub}} \label{app:proof-thm-F1H-ub}

For UEs $\ell\in\cI_A$, $\pi_H$ (Algorithm \ref{algo:policy}) increments counter $b_\ell$ (adds packet to set $\cS$) using a threshold criterion. Due to the threshold criterion, the inter-arrival time of packets transmitted to UE $\ell$ is $T_{\ell i}^H=\lceil\avT_\ell^\star-1/q_\ell\rceil+X_\ell$, where $X_\ell$ is a geometrically distributed random variable with mean $1/q_\ell$ that represents the number of slots (relative to a fixed slot) until a packet arrives at the BS for UE $\ell$. Therefore, $T_{\ell i}^H$'s, for all transmitted packets $i$, are independent and identically distributed with mean $\bbE[T_{\ell i}^H]\in[\avT_\ell^\star,\avT_\ell^\star+1)$ and variance $Var(T_{\ell i}^H)=(1-q_\ell)/q_\ell^2$. Hence, using the strong law of large numbers, with probability $1$, we get $\avT_\ell^H\in[\avT_\ell^\star, \avT_\ell^\star+1)$, and  
	$$\overline{(\delta_\ell^H)^2}=\lim_{t\to\infty}\sum_{i\in\cR_\ell^H(t)}\frac{(T_{\ell i}^H-\avT_\ell^H)^2}{|\cR_\ell^H(t)|}= \frac{(1-q_\ell)}{q_\ell^2}, $$ 
$\forall \ell\in\cI_A$. 	
Further, substituting these results in \eqref{eq:objective-equiv-term1}, we get 
\begin{align*}
	F_1^H &< \sum_{\ell\in\cI_A}\frac{\rho_\ell}{2}\bigg[(\avT_\ell^\star+1)+\frac{(1-q_\ell)/q_\ell^2}{\avT_\ell^\star}+1\bigg], \nonumber \\
	&\stackrel{(a)}{=}\min_{\avT_\ell, \forall \ell\in\cI_A} \sum_{\ell\in\cI_A}\frac{\rho_\ell}{2}\bigg[\avT_\ell+\frac{(1-q_\ell)/q_\ell^2}{\avT_\ell}+2\bigg]\ \ \text{s.t. \eqref{eq:avT-constraint}}, \nonumber \\
	&\stackrel{(b)}{<}2\cdot \min_{\pi\in\Pi_F} \sum_{\ell\in\cI_A}\frac{\rho_\ell}{2}\bigg[\avT_\ell^\pi+\frac{(1-q_\ell)/2q_\ell^2}{\avT_\ell^\pi}+1\bigg], \nonumber \\
	&\stackrel{(c)}{=}2\cdot \min_{\pi\in\Pi_F} F_1^\pi, 
\end{align*}
where $(a)$ follows from the definition of $T_\ell^\star$ \eqref{eq:compute-Tstar}, $(b)$ follows because $\pi\in\Pi_F$ satisfy \eqref{eq:avT-constraint} (Lemma \ref{lemma:throughput-UEs}), and we get $(c)$ using \eqref{eq:F1-lb}.

\section{} 
\label{app:proof-lemma-throughput-piH}

\begin{lemma} \label{lemma:throughput-piH}
For policy $\pi_H$ (Algorithm \ref{algo:policy}),  
\begin{align} \label{eq:throughput-piH-dummy}
	\lim_{t\to\infty}\sum_{k\in\cI_R}\hat{R}_k^H(t)/t\ge \sum_{k\in\cI_R}\alpha_k/p_k,
\end{align}
where $\hat{R}_k^H(t)$ denotes the number of slots until slot $t$ in which BS attempts to transmit to UE $k$.
\end{lemma}
\begin{proof}
Note that $\pi_H$ transmits in every slot (when set $\cS$ is empty, $\pi_H$ transmits packets of UEs $k\in\cI_R$ that always have a packet to transmit; Remark \ref{remark:q=1}).
Therefore, $\sum_{i\in\cI_A}\hat{R}_i^H(t)+\sum_{j\in\cI_L}\hat{R}_j^H(t)+\sum_{k\in\cI_R}\hat{R}_k^H(t)= t$. Dividing both sides by $t$, we get
\begin{align} \label{eq:throughput-piH-1}
	\lim_{t\to\infty}\sum_{k\in\cI_R}\frac{\hat{R}_k^H(t)}{t}&=1-\lim_{t\to\infty}\bigg(\sum_{i\in\cI_A}\frac{\hat{R}_i^H(t)}{t}+\sum_{j\in\cI_L}\frac{\hat{R}_j^H(t)}{t}\bigg), \nonumber \\
	&\stackrel{(a)}{=}1-\left(\sum_{i\in\cI_A}\frac{\avR_i^H}{p_i}+\sum_{j\in\cI_L}\frac{\avR_j^H}{p_j}\right), \nonumber \\
	&\stackrel{(b)}{\ge} 1-\left(\sum_{i\in\cI_A}\frac{\avR_i^H}{p_i}+\sum_{j\in\cI_L}\frac{q_j}{p_j}\right),
\end{align}
where we get $(a)$ using \eqref{eq:throughput-equiv}, and $(b)$ follows because the throughput $\avR_j^H\le q_j$ (for UE $j\in\cI_L$, throughput cannot exceed the packet arrival rate). 

Further, since $\avT_i^\star$ satisfies \eqref{eq:avT-constraint}, we get 
\begin{align} \label{eq:throughput-piH-2}
	1-\sum_{j\in\cI_L}\frac{q_j}{p_j}-\sum_{k\in\cI_R}\frac{\alpha_k}{p_k}\ge \sum_{i\in\cI_A}\frac{1}{p_i\avT_i^\star} \stackrel{(a)}{\ge} \sum_{i\in\cI_A}\frac{\avR_i^H}{p_i},
\end{align}
where $(a)$ follows because i) $\avR_i^H=1/\avT_i^H$ (from Lemma \ref{lemma:relation-throughput-cycle-length}), and ii) $\avT_i^H\ge \avT_i^\star$ (shown in Theorem \ref{thm:F1H-ub}). Substituting \eqref{eq:throughput-piH-2} into \eqref{eq:throughput-piH-1}, we get \eqref{eq:throughput-piH-dummy}. 
\end{proof}


\bibliographystyle{IEEEtran}
\bibliography{Reflist}

\begin{thebibliography}{10}
\providecommand{\url}[1]{#1}
\csname url@samestyle\endcsname
\providecommand{\newblock}{\relax}
\providecommand{\bibinfo}[2]{#2}
\providecommand{\BIBentrySTDinterwordspacing}{\spaceskip=0pt\relax}
\providecommand{\BIBentryALTinterwordstretchfactor}{4}
\providecommand{\BIBentryALTinterwordspacing}{\spaceskip=\fontdimen2\font plus
\BIBentryALTinterwordstretchfactor\fontdimen3\font minus
  \fontdimen4\font\relax}
\providecommand{\BIBforeignlanguage}[2]{{%
\expandafter\ifx\csname l@#1\endcsname\relax
\typeout{** WARNING: IEEEtran.bst: No hyphenation pattern has been}%
\typeout{** loaded for the language `#1'. Using the pattern for}%
\typeout{** the default language instead.}%
\else
\language=\csname l@#1\endcsname
\fi
#2}}
\providecommand{\BIBdecl}{\relax}
\BIBdecl

\bibitem{pantelidou2011scheduling}
A.~Pantelidou, A.~Ephremides \emph{et~al.}, ``Scheduling in wireless
  networks,'' \emph{Foundations and Trends{\textregistered} in Networking},
  vol.~4, no.~4, pp. 421--511, 2011.

\bibitem{hou2009theory}
I.-H. Hou, V.~Borkar, and P.~Kumar, \emph{A theory of QoS for wireless}.\hskip
  1em plus 0.5em minus 0.4em\relax IEEE, 2009.

\bibitem{lee2007reverse}
J.-W. Lee, A.~Tang, J.~Huang, M.~Chiang, and A.~R. Calderbank,
  ``Reverse-engineering mac: A non-cooperative game model,'' \emph{IEEE Journal
  on Selected Areas in Communications}, vol.~25, no.~6, pp. 1135--1147, 2007.

\bibitem{buyukkoc1985cmu}
C.~Buyukkoc, P.~Varaiya, and J.~Walrand, ``The c$\mu$ rule revisited,''
  \emph{Advances in applied probability}, vol.~17, no.~1, pp. 237--238, 1985.

\bibitem{saurav2022scheduling}
K.~Saurav and R.~Vaze, ``Scheduling to minimize age of information with
  multiple sources,'' in \emph{2022 20th International Symposium on Modeling
  and Optimization in Mobile, Ad hoc, and Wireless Networks (WiOpt)}, 2022, pp.
  57--64.

\bibitem{kaul2012status}
S.~K. Kaul, R.~D. Yates, and M.~Gruteser, ``Status updates through queues,'' in
  \emph{2012 46th Annual conference on information sciences and systems
  (CISS)}.\hskip 1em plus 0.5em minus 0.4em\relax IEEE, 2012, pp. 1--6.

\bibitem{mirzaeinnia2020latency}
A.~Mirzaeinnia, M.~Mirzaeinia, and A.~Rezgui, ``Latency and throughput
  optimization in modern networks: A comprehensive survey,'' \emph{arXiv
  preprint arXiv:2009.03715}, 2020.

\bibitem{saxena2018review}
P.~Saxena and P.~Patidar, ``Review paper on throughput optimization and
  spectrum sensing in cognitive radio,'' \emph{International Research Journal
  of Engineering and Technology}, vol.~5, no.~8, pp. 602--605, 2018.

\bibitem{kushner2004convergence}
H.~J. Kushner and P.~A. Whiting, ``Convergence of proportional-fair sharing
  algorithms under general conditions,'' \emph{IEEE transactions on wireless
  communications}, vol.~3, no.~4, pp. 1250--1259, 2004.

\bibitem{sun2017remote}
Y.~Sun, Y.~Polyanskiy, and E.~Uysal-Biyikoglu, ``Remote estimation of the
  wiener process over a channel with random delay,'' in \emph{2017 IEEE
  International Symposium on Information Theory (ISIT)}.\hskip 1em plus 0.5em
  minus 0.4em\relax IEEE, 2017, pp. 321--325.

\bibitem{ayan2019age}
O.~Ayan, M.~Vilgelm, M.~Kl{\"u}gel, S.~Hirche, and W.~Kellerer,
  ``Age-of-information vs. value-of-information scheduling for cellular
  networked control systems,'' in \emph{Proceedings of the 10th ACM/IEEE
  International Conference on Cyber-Physical Systems}, 2019, pp. 109--117.

\bibitem{maatouk2022age}
A.~Maatouk, M.~Assaad, and A.~Ephremides, ``The age of incorrect information:
  An enabler of semantics-empowered communication,'' \emph{IEEE Transactions on
  Wireless Communications}, 2022.

\bibitem{joshi2021minimization}
B.~Joshi, R.~V. Bhat, B.~Bharath, and R.~Vaze, ``Minimization of age of
  incorrect estimates of autoregressive markov processes,'' in \emph{2021 19th
  International Symposium on Modeling and Optimization in Mobile, Ad hoc, and
  Wireless Networks (WiOpt)}.\hskip 1em plus 0.5em minus 0.4em\relax IEEE,
  2021, pp. 1--8.

\bibitem{kadota2018optimizing}
I.~Kadota, A.~Sinha, and E.~Modiano, ``Optimizing age of information in
  wireless networks with throughput constraints,'' in \emph{IEEE INFOCOM
  2018-IEEE Conference on Computer Communications}.\hskip 1em plus 0.5em minus
  0.4em\relax IEEE, 2018, pp. 1844--1852.

\bibitem{kosta2018age}
A.~Kosta, N.~Pappas, A.~Ephremides, and V.~Angelakis, ``Age of information and
  throughput in a shared access network with heterogeneous traffic,'' in
  \emph{2018 IEEE Global Communications Conference (GLOBECOM)}.\hskip 1em plus
  0.5em minus 0.4em\relax IEEE, 2018, pp. 1--6.

\bibitem{saurav2021minimizing}
K.~Saurav and R.~Vaze, ``Minimizing the sum of age of information and
  transmission cost under stochastic arrival model,'' in \emph{IEEE INFOCOM
  2021-IEEE Conference on Computer Communications}.\hskip 1em plus 0.5em minus
  0.4em\relax IEEE, 2021, pp. 1--10.

\bibitem{gopal2019coexistence}
S.~Gopal, S.~K. Kaul, and R.~Chaturvedi, ``Coexistence of age and throughput
  optimizing networks: A game theoretic approach,'' in \emph{2019 IEEE 30th
  Annual International Symposium on Personal, Indoor and Mobile Radio
  Communications (PIMRC)}.\hskip 1em plus 0.5em minus 0.4em\relax IEEE, 2019,
  pp. 1--6.

\bibitem{gopal2018game}
S.~Gopal and S.~K. Kaul, ``A game theoretic approach to dsrc and wifi
  coexistence,'' in \emph{IEEE INFOCOM 2018-IEEE Conference on Computer
  Communications Workshops (INFOCOM WKSHPS)}.\hskip 1em plus 0.5em minus
  0.4em\relax IEEE, 2018, pp. 565--570.

\bibitem{bhat2020throughput}
R.~V. Bhat, R.~Vaze, and M.~Motani, ``Throughput maximization with an average
  age of information constraint in fading channels,'' \emph{IEEE Transactions
  on Wireless Communications}, vol.~20, no.~1, pp. 481--494, 2020.

\bibitem{al2019information}
A.~O. Al-Abbasi, A.~Elghariani, A.~Elgabli, and V.~Aggarwal, ``On the
  information freshness and tail latency trade-off in mobile networks,'' in
  \emph{2019 IEEE Global Communications Conference (GLOBECOM)}.\hskip 1em plus
  0.5em minus 0.4em\relax IEEE, 2019, pp. 1--6.

\bibitem{sun2021age}
J.~Sun, L.~Wang, Z.~Jiang, S.~Zhou, and Z.~Niu, ``Age-optimal scheduling for
  heterogeneous traffic with timely throughput constraints,'' \emph{IEEE
  Journal on Selected Areas in Communications}, vol.~39, no.~5, pp. 1485--1498,
  2021.

\bibitem{fountoulakis2023scheduling}
E.~Fountoulakis, T.~Charalambous, A.~Ephremides, and N.~Pappas, ``Scheduling
  policies for aoi minimization with timely throughput constraints,''
  \emph{IEEE Transactions on Communications}, 2023.

\bibitem{krishnasamy2018learning}
S.~Krishnasamy, A.~Arapostathis, R.~Johari, and S.~Shakkottai, ``On learning
  the c$\mu$ rule in single and parallel server networks,'' in \emph{2018 56th
  Annual Allerton Conference on Communication, Control, and Computing
  (Allerton)}.\hskip 1em plus 0.5em minus 0.4em\relax IEEE, 2018, pp. 153--154.

\bibitem{srikant2013communication}
R.~Srikant and L.~Ying, \emph{Communication networks: an optimization, control,
  and stochastic networks perspective}.\hskip 1em plus 0.5em minus 0.4em\relax
  Cambridge University Press, 2013.

\bibitem{kadota2019minimizing}
I.~Kadota and E.~Modiano, ``Minimizing the age of information in wireless
  networks with stochastic arrivals,'' \emph{IEEE Transactions on Mobile
  Computing}, vol.~20, no.~3, pp. 1173--1185, 2019.

\end{thebibliography}

\end{document}